\documentclass[10pt,draft,onecolumn]{IEEEtran}
\usepackage{epsf,psfrag,amssymb,amsfonts,cite,subfigure}
\newtheorem{theorem}{Theorem}
\newtheorem{example}{Example}
\newtheorem{corollary}{Corollary}
\newtheorem{lemma}{Lemma}
\newtheorem{definition}{Definition}

\def\psfancypar#1#2{\begingroup\def\par{\endgraf\endgroup\lineskiplimit=0pt}
               \setbox2=\hbox{\large\sc #2}
               \newdimen\tmpht \tmpht \ht2 \advance\tmpht by \baselineskip
               \font\hhuge=Times-Bold at \tmpht
               \setbox1=\hbox{{\hhuge #1}}
               \count7=\tmpht \count8=\ht1
               \divide\count8 by 1000 \divide\count7 by \count8 
               \tmpht=.001\tmpht\multiply\tmpht by \count7 
               \font\hhuge=Times-Bold at \tmpht
               \setbox1=\hbox{{\hhuge #1}}
               \noindent
                \hangindent1.05\wd1
               \hangafter=-2 {\hskip-\hangindent
               \lower1\ht1\hbox{\raise1.0\ht2\copy1}%
                \kern-0\wd1}\copy2\lineskiplimit=-1000pt}

\newcommand{\beq}{\begin{equation}}
\newcommand{\eeq}{\end{equation}}
\newcommand{\bqa}{\begin{eqnarray}}
\newcommand{\eqa}{\end{eqnarray}}
\newcommand{\bqn}{\begin{eqnarray*}}
\newcommand{\eqn}{\end{eqnarray*}}

\newcommand{\be}{\begin{enumerate}}
\newcommand{\ee}{\end{enumerate}}
\newcommand{\bi}{\begin{itemize}}
\newcommand{\ei}{\end{itemize}}
\newcommand{\bd}{\begin{description}}
\newcommand{\ed}{\end{description}}
\newcommand{\ba}{\begin{array}}
\newcommand{\ea}{\end{array}}
\newcommand{\bde}{\begin{definition}}
\newcommand{\ede}{\end{definition}}
\newcommand{\bex}{\begin{example}}
\newcommand{\eex}{\end{example}}

\def\boxit#1{\vbox{\hrule\hbox{\vrule\kern3pt
        \vbox{\kern3pt#1\kern3pt}\kern3pt\vrule}\hrule}}

\def\reals{ { {\rm  I \kern-0.15em R }  } }
\def\complex{ {\,{{\rm C} \kern-0.50em \raise0.20ex {  |}}\, }}

\def\0bf{{\bf 0}}
\def\1bf{{\bf 1}}
\def\2bf{{\bf 2}}
\def\3bf{{\bf 3}}
\def\4bf{{\bf 4}}
\def\5bf{{\bf 5}}
\def\6bf{{\bf 6}}
\def\7bf{{\bf 7}}
\def\8bf{{\bf 8}}
\def\9bf{{\bf 9}}

\def\Rbf{{\bf R}}

\def\Nmat{\mathcal{N}}

\def\Qmat{\mathcal{Q}}

\def\Xmat{\mathcal{X}}
\def\Ymat{\mathcal{Y}}

\def\Rxx{\Rbf_{\ssstyle X\kern-.1em X}}

\let\ssstyle=\scriptscriptstyle

\def\Kout{\setbox1=\hbox{\Huge\bf K}\hbox to
1.05\wd1{\hspace{.05\wd1}
}}
\def\Sout{\setbox1=\hbox{\Huge\bf S}\hbox to 1.05\wd1{\hspace{.05\wd1}
}}

\def\scalefig#1{\epsfxsize #1\textwidth}
\begin{document}
\title{\LARGE On the Capacity of Multiple-Access-Z-Interference Channels}
\author{Fangfang Zhu$^*$, Xiaohu Shang$^\dagger$, Biao Chen$^*$, H. Vincent Poor$^\ddag$\footnotemark\\
$^*$ Dept. of EECS, Syracuse University, Syracuse, NY\\
$^\dagger$ Bell Labs, Alcatel-Lucent, Holmdel, NJ\\
$^\ddag$ Dept. of Electrical Engineering, Princeton University, Princeton, NJ}


\maketitle
\let\thefootnote\relax\footnotetext{This work was supported in part by the National Science Foundation under Grants CCF-0905320, and CNS-0905398, and in part by the Air Force Office of Scientific Research under Grant FA9550-09-01-0643. The material in this paper was presented in part at the IEEE International Conference on Communications (ICC), Kyoto, Japan, June 2011.}
\begin{abstract}
The capacity of a network in which a multiple access channel (MAC) generates interference to a single-user channel is studied. An achievable rate region based on superposition coding and joint decoding is established for the discrete case. If the interference is very strong, the capacity region is obtained for both the discrete memoryless channel and the Gaussian channel. For the strong interference case, the capacity region is established for the discrete memoryless channel; for the Gaussian case, we attain a line segment on the boundary of the capacity region. Moreover, the capacity region for the Gaussian channel is identified for the case when one interference link being strong, and the other being very strong. For a subclass of Gaussian channels with mixed interference, a boundary point of the capacity region is determined. Finally, for the Gaussian channel with weak interference, sum capacities are obtained under various channel coefficient and power constraint conditions.
\end{abstract}

\section{Introduction}
In a cellular system, co-channel cells are strategically placed to ensure that interference is kept at a minimum.
As such, the downlink transmission within each cell is typically modeled as a broadcast channel (BC) while uplink transmission is
modeled as a multiple access channel (MAC). This effectively isolates each cell from all the other co-channel cells and makes it feasible to characterize the performance limits as the capacity regions for the Gaussian BC and the Gaussian MAC have been completely determined  (see \cite{Cover&Thomas:book}).

However, as the need for spectrum reuse increases, various frequency reuse schemes have been proposed in recent years and it is no longer realistic
to disregard co-channel interference in both downlink and uplink transmissions.
For downlink transmissions, the Gaussian broadcast-interference channel model has been studied in \cite{Shang&Poor: 10ISIT, Shang&Poor:10IT, Liu&Erkip:12IT} with
an emphasis on the one-sided interference model. The capacity regions of such channels with very strong and slightly strong interference,
and some boundary points on the capacity regions of that with moderate and weak interferences were determined.
It was shown that the capacity is achieved by fully decoding the interference when it is strong,
partially decoding the interference when it is moderate, and treating the interference as noise when it is weak.

In this paper, we consider an uplink model with interference, namely the multiple access-interference channel.
As with \cite{Shang&Poor: 10ISIT, Shang&Poor:10IT, Liu&Erkip:12IT}, we focus on the MAC with one-sided interference, an example of this channel model is depicted is illustrated in Fig.~\ref{fig:cellmodel}. The same model can be used to describe the channels between microcell and femtocell, or between microcell and picocell, etc. 
Mobile users $TX_1$ and $TX_2$ belong to cell $1$ while $TX_3$ belongs to cell 2 and the transmissions of $TX_1$ and $TX_2$ cause
interferences at $RX_2$, the base station at cell $2$. The interference from $TX_3$ to $RX_1$, on the other hand, is assumed to be
negligible.

A similar model has been studied by \cite{Chaaban&Sezgin:EW11} and \cite{Buhler&Wunder:ISIT2011}, both of which considered the two-sided interference between the two cells. The authors in \cite{Chaaban&Sezgin:EW11} derived the capacity region for the very strong and some of the strong interference cases, and provided an upper-bound of the sum-rate for the weak interference case which is nearly optimal in low signal-to-noise ratio regime, while \cite{Buhler&Wunder:ISIT2011} characterized the capacity region in the from of interference alignment under the weak symmetric interference assumption.
\begin{figure}[h]
\centerline{
\begin{psfrags}\psfrag{r1}[c]{\tiny$Rx_1$} \psfrag{r2}[c]{\tiny$Rx_2$}\psfrag{t1}[c]{\tiny$Tx_1$}\psfrag{t2}[c]{\tiny$Tx_2$}\psfrag{t3}[c]{\tiny$Tx_3$}\scalefig{0.6}
\epsfbox{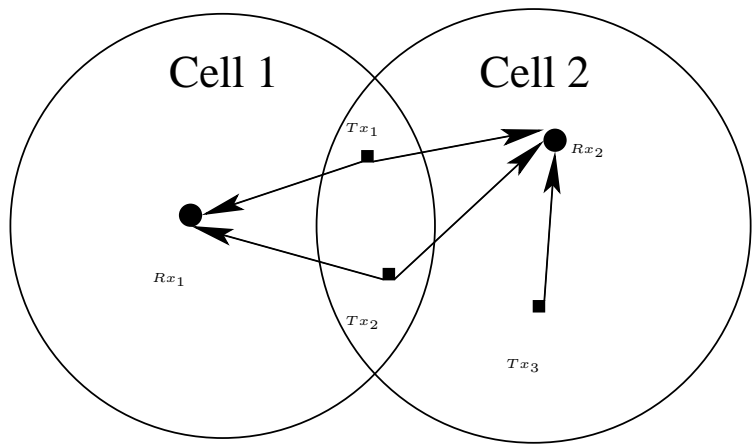}\end{psfrags}}
\caption{\label{fig:cellmodel} Two-cell uplink transmission.}
\end{figure}

Fig.~\ref{fig:MAZICmodel} is an abstract model of the above network. Transmitters $1$ and $2$ and receiver $1$ form a MAC.
Transmitter $3$ and receiver $2$ form a single-user channel and receiver $2$ is subject to interference from transmitters $1$ and $2$. Specifically, the channel outputs are given by

\bqa
Y_1 &=& X_1 + X_2 + Z_1,\\
Y_2 &=& \sqrt{a}X_1 + \sqrt{b}X_2+X_3+Z_2,
\eqa
\begin{figure}[h]
\centerline{
\begin{psfrags}\psfrag{h11}[c]{\small$1$}\psfrag{h12}[c]{\small$\sqrt{a}$}\psfrag{h21}[c]{\small$1$}\psfrag{h22}[c]{\small$\sqrt{b}$}\psfrag{h32}[c]{\small$1$}\psfrag{x1}[c]{\small$X_1(W_1)$}\psfrag{x2}[c]{\small$X_2(W_2)$}\psfrag{x3}[c]{\small$X_3(W_3)$}
\psfrag{y1}[l]{\small$Y_1(W_1,W_2)$}\psfrag{y2}[l]{\small$Y_2(W_3)$}
\scalefig{0.3}\epsfbox{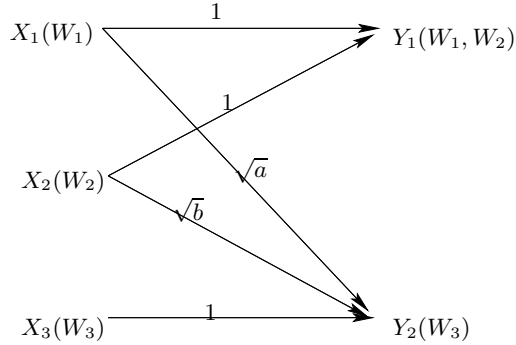}\end{psfrags}}
\caption{\label{fig:MAZICmodel} The Multiple-Access-Z-interference
Channel model}
\end{figure}
where $X_i$ and $Y_j$ are the transmitted and received signals of transmitter $i$ and receiver $j$,
respectively, for $i=1,2,3$ and $j=1,2$. For each $j$, $Z_j$ is Gaussian noise with zero mean and unit variance and we
assume all the noise terms are independent of each other and over time. For channels with arbitrary coefficients and noise variances, standard normalization can be applied such that its capacity is equivalent to the above channel, i.e., the gains for $X_1$, $X_2$ in $Y_1$ and $X_3$ in $Y_2$ are all assumed to be $1$. The channel coefficients $a$ and $b$ are fixed and known at both the transmitters and the receivers.
Without loss of generality, we assume $a, b > 0$, i.e., they are strictly positive. 
For transmitter $i$, the user/channel input sequence $X_{i1},X_{i2}, \cdots,X_{in}$ is subject to a block power constraint
$\sum_{k=1}^n \mathcal{E}[X_{ik}^2] \leq nP_{i}$. We denote the rates for messages $W_1$, $W_2$ and $W_3$ by $R_1$, $R_2$ and $R_3$, respectively.
The channel defined here is referred to as a Multiple-Access-Z-Interference channel (MAZIC). Unlike the two-user Z-interference channel (ZIC), there are more than one interference signal from multiple independent senders. For example, in the Gaussian case, the interference signals are multiplied by different coefficients. One cannot claim equivalence to degraded channels as in the two-user ZIC case. As such, capacity analysis becomes more complicated.
Our goal in this paper is to
obtain capacity results for the strong, mixed$^1$\footnote{$^1$Here, the notion of mixed interference refers to the strengths of the two interference links with coefficients $\sqrt{a}$ and $\sqrt{b}$. It differs from the classical notion of mixed interference where the interference is imposed on two different receivers.}  and weak interference cases for the MAZIC.

The rest of the paper is organized as follows. We give the problem formulation in Section \ref{sec:model}.
Section \ref{sec:achievable} gives an achievable rate region for the discrete memoryless MAZIC and the result is extended to the Gaussian case.
Capacity results for the strong, very strong, mixed and weak interference cases are derived in Sections \ref{sec: strong}, \ref{sec:verystrong}, \ref{sec:mixed} and \ref{sec:weak} respectively.
Section \ref{sec:conclusion} concludes the paper.

\section{Preliminaries}\label{sec:model}
A discrete memoryless MAZIC is defined by $(\Xmat_1,\Xmat_2,\Xmat_3, p,\Ymat_1,\Ymat_2)$, where $\Xmat_1, \Xmat_2$ and $\Xmat_3$ are finite input
alphabet sets; $\Ymat_1$ and $\Ymat_2$ are finite output alphabet sets; and $p (y_1y_2|x_1x_2x_3 )$ is the channel transition probability.
As the receivers do not cooperate, the capacity depends only on the marginal channel transition probabilities. Thus we can only consider two marginal distributions $(p(y_1|x_1x_2), p(y_2|x_1x_2x_3))$.
The channels are memoryless, i.e.,
\bqa
p(y_1^ny_2^n|x_1^nx_2^nx_3^n)=\prod_{i=1}^n p(y_{1i}y_{2i}|x_{1i}x_{2i}x_{3i}),
\eqa
where $x_i^n=[x_{i1},x_{i2},\cdots,x_{in}]$ and $y_j^n=[y_{j1},y_{j2},\cdots,y_{jn}]$, for $i=1,2$, and $j=1,2,3$.
The message for transmitter $i$ is $W_i\in\{1, 2, \cdots, 2^{nR_i}\}$, $i=1,2,3$. A $(2^{nR_1}, 2^{nR_2}, 2^{nR_3},n)$ code consists of three encoders:
\bqn
f_1&:&\{1,2,\cdots, 2^{nR_1}\}\rightarrow\Xmat_1^n,\\
f_2&:&\{1,2,\cdots, 2^{nR_2}\}\rightarrow\Xmat_2^n,\\
f_3&:&\{1,2,\cdots, 2^{nR_3}\}\rightarrow\Xmat_3^n,
\eqn
and two decoders:
\bqn
g_1&:&\Ymat_1^n\rightarrow\{1,2,\cdots,2^{nR_1}\}\times\{1,2,\cdots,2^{nR_2}\},\\
g_2&:&\Ymat_2^n\rightarrow\{1,2,\cdots,2^{nR_3}\}.
\eqn
The error probability is defined as
\bqn
P_e=\Pr\{g_1(Y_1^n)\neq (W_1,W_2),\textrm{ or }g_2(Y_2^n)\neq W_3 \}.
\eqn
Assuming $W_1$, $W_2$ and $W_3$ are all uniformly distributed, a rate triple $(R_1,R_2,R_3)$ is achievable if there exist a sequence of
$(2^{nR_1}, 2^{nR_2}, 2^{nR_3}, n)$ codes for $n$ sufficiently large such that $P_e\rightarrow 0$ when $n\rightarrow \infty$.
Throughout this paper, we make the assumption that all the transmitters implement deterministic encoders instead of stochastic encoders
as one can easily prove, following the same approach as that of \cite{Willems&Meulen:85IT}, that stochastic encoders do not increase the capacity
for a MAZIC. 
Before proceeding, we introduce some notation
that will be used throughout the paper.
\bi
\item \noindent$p_X(x)$ is the probability mass function of a discrete random variable $X$, or a probability density function of a continuous
random variable $X$, and is simplified as $p(x)$. 
\item $A_\epsilon^{(n)}(X)$ denotes the set of length-$n$ $\epsilon$-typical sequences of $X$.
\item $I(\cdot ;\cdot)$, $H(\cdot)$ and $h(\cdot)$ are respectively the mutual information, discrete entropy and differential entropy.
\item $\emptyset$ denotes the empty set.
\item $\bar{x}=1-x$.
\item $\bf x\sim \Nmat(\bf 0,S)$ means that $\bf x$ has a Gaussian distribution with zero mean and covariance matrix $\bf S$.
\ei

The following properties of Markov chains are useful throughout the paper (see \cite[Section 1.1.5]{Pearl:book}):
\begin{itemize}
  \item Decomposition: $X-Y-ZW\Longrightarrow X-Y-Z$;
  \item Weak Union: $X-Y-ZW\Longrightarrow X-YW-Z$;
  \item Contraction: $(X-Y-Z)$ and $(X-YZ-W)\Longrightarrow X-Y-ZW$.
\end{itemize}

\section{An Achievable Region for the General MAZIC}\label{sec:achievable}
We use superposition coding and joint decoding to derive an achievable rate region. Consider the
independent messages $W_1$ and $W_2$ generated by transmitters $1$ and $2$, respectively. We split them into
\bqn
W_1&=&[W_{1c},W_{1p}],\\
W_2&=&[W_{2c},W_{2p}],
\eqn
where $W_{1c}$ and $W_{2c}$ denote the common messages that are to be decoded at both receivers $1$ and $2$; and $W_{1p}$ and $W_{2p}$ represent
the private messages that are to be decoded only at receiver $1$. 

We first introduce the auxiliary random variables $Q$, $U_1$, and $U_2$, where $Q$ is a time-sharing random variable, and $U_1$ and $U_2$ contain the information $W_{1c}$ and $W_{2c} $ respectively. The distribution of $(
Q, U_1, U_2, X_1, X_2, X_3)$ factorizes as
\bqa
p(qu_1u_2x_1x_2x_3)=p(q)p(u_1|q)p(x_1|u_1,q)p(u_2|q)p(x_2|u_2,q)p(x_3|q).\label{eq:input_dist}
\eqa
The following achievable rate region can be obtained whose proof is given in Appendix \ref{apx:general_inner}.
\vspace{10pt}
\begin{theorem}\label{thm:general_inner}
For a discrete memoryless MAZIC, an achievable rate region is given
by the set of all nonnegative rate triples $(R_1,R_2,R_3)$ that satisfy
\bqa
R_1\!\!\!\!&\leq&\!\!\!\!I(X_1;Y_1|X_2Q),\label{eq:thm1_start}\\
R_2\!\!\!\!&\leq&\!\!\!\!I(X_2;Y_1|X_1Q),\\
R_3\!\!\!\!&\leq&\!\!\!\!I(X_3;Y_2|U_1U_2Q),\\
R_1+R_2\!\!\!\!&\leq&\!\!\!\!I(X_1X_2;Y_1|Q),\\
R_1+R_3\!\!\!\!&\leq&\!\!\!\!I(X_1;Y_1|U_1X_2Q)+I(U_1X_3;Y_2|U_2Q),\\
R_2+R_3\!\!\!\!&\leq&\!\!\!\!I(X_2;Y_1|U_2X_1Q)+I(U_2X_3;Y_2|U_1Q),\\
R_1+R_2+R_3\!\!\!\!&\leq&\!\!\!\!I(X_1X_2;Y_1|U_1U_2Q)+I(U_1U_2X_3;Y_2|Q),\\
R_1+R_2+R_3\!\!\!\!&\leq&\!\!\!\!I(X_1X_2;Y_1|U_1Q)+I(U_1X_3;Y_2|U_2Q),\\
R_1+R_2+R_3\!\!\!\!&\leq&\!\!\!\!I(X_1X_2;Y_1|U_2Q)+I(U_2X_3;Y_2|U_1Q),\\
R_1+2R_2+R_3\!\!\!\!&\leq&\!\!\!\!I(X_2;Y_1|U_2X_1Q)+I(X_1X_2;Y_1|U_1Q)+I(U_1U_2X_3;Y_2|Q),\\
2R_1+R_2+R_3\!\!\!\!&\leq&\!\!\!\!I(X_1;Y_1|U_1X_2Q)+I(X_1X_2;Y_1|U_2Q)+I(U_1U_2X_3;Y_2|Q),\label{eq:thm1_end}
\eqa
\normalsize
where the input distribution factors as (\ref{eq:input_dist}).
Furthermore, the region remains the same if we impose the constraints $\|\Qmat\|\leq 12$, $\|U_1\|\leq \|X_1\|+5$, and $\|U_2\|\leq \|X_2\|+5$.
\end{theorem}
\vspace{10pt}
The MAC and the Z-interference channel (ZIC) are two special cases of a MAZIC. On setting $X_3U_1U_2=\emptyset$, we obtain the capacity region for the MAC:
\bqn
R_1&\leq& I(X_1;Y_1|X_2Q),\\
R_2&\leq& I(X_2;Y_1|X_1Q),\\
R_1+R_2&\leq& I(X_1X_2;Y_1|Q).
\eqn
Alternatively, on setting $U_2X_2=\emptyset$, we obtain Han and Kobayashi's achievable rate region for the ZIC \cite{Han&Kobayashi:81IT}\cite{Chong-etal:08IT}\cite{Kramer:06IZS}:
\bqn
R_1 &\leq& I(X_1;Y_1|Q),\\
R_3 &\leq& I(X_3;Y_2|U_1Q),\\
R_1+R_3 &\leq& I(X_1;Y_1|U_1Q)+I(U_1X_3;Y_2|Q).
\eqn
Theorem \ref{thm:general_inner} allows us to obtain a computable achievable region for Gaussian MAZICs.
\vspace{10pt}
\begin{corollary}\label{cor:general_g_inner}
For any nonnegative pair $[\alpha, \beta]\in[0,1]$, the non-negative rate triples $(R_1,R_2,R_3)$ satisfying the conditions (\ref{eq: cor1_start})-(\ref{eq:cor1_end}) are achievable for a Gaussian MAZIC.

\bqa
R_1 &\leq& \frac{1}{2}\log(1+P_1),\label{eq: cor1_start}\\
R_2 &\leq& \frac{1}{2}\log(1+P_2),\\
R_3 &\leq& \frac{1}{2}\log\left(1+\frac{P_3}{1+a\alpha P_1+b\beta P_2}\right),\\
R_1+R_2 &\leq& \frac{1}{2}\log\left(1+P_1+P_2\right),\\
R_1+R_3 &\leq& \frac{1}{2}\log\left(1+\alpha P_1\right)+\frac{1}{2}\log\left(1+\frac{a\bar{\alpha}P_1+P_3}{1+a\alpha P_1+b\beta P_2}\right),\\
R_2+R_3 &\leq& \frac{1}{2}\log\left(1+\beta P_2\right)+\frac{1}{2}\log\left(1+\frac{b\bar{\beta}P_2+P_3}{1+a\alpha P_1+b\beta P_2}\right),\\
R_1+R_2+R_3 &\leq& \frac{1}{2}\log\left(1+\alpha P_1+\beta P_2\right)+\frac{1}{2}\log\left(1+\frac{a\bar{\alpha}P_1+b\bar{\beta}P_2+P_3}{1+a\alpha P_1+b\beta P_2}\right),\\
R_1+R_2+R_3 &\leq& \frac{1}{2}\log\left(1+\alpha P_1+P_2\right)+\frac{1}{2}\log\left(1+\frac{a\bar{\alpha}P_1+P_3}{1+a\alpha P_1+b\beta P_2}\right),\\
R_1+R_2+R_3 &\leq& \frac{1}{2}\log\left(1+P_1+\beta P_2\right)+\frac{1}{2}\log\left(1+\frac{b\bar{\beta}P_2+P_3}{1+a\alpha P_1+b\beta P_2}\right),\\
R_1+2R_2+R_3 &\leq& \frac{1}{2}\log\left(1+\beta P_2\right)+\frac{1}{2}\log\left(1+\alpha P_1+P_2\right)+\frac{1}{2}\log\left(1+\frac{a\bar{\alpha}P_1+b\bar{\beta}P_2+P_3}{1+a\alpha P_1+b\beta P_2}\right),\\
2R_1+R_2+R_3&\leq& \frac{1}{2}\log\left(1+\alpha P_1\right)+\frac{1}{2}\log\left(1+P_1+\beta P_2\right)+\frac{1}{2}\log\left(1+\frac{a\bar{\alpha}P_1+b\bar{\beta}P_2+P_3}{1+a\alpha P_1+b\beta P_2}\right).\label{eq:cor1_end}
\eqa
\end{corollary}
\vspace{10pt}
\begin{proof}
Corollary \ref{cor:general_g_inner} follows directly from Theorem \ref{thm:general_inner} by choosing $\|\Qmat\|=1$, $X_1\sim\Nmat(0,P_1)$, $X_2\sim \Nmat(0,P_2)$, and $X_1=U_1+V_1$, $X_2=U_2+V_2$, where $U_1$, $U_2$, $V_1$ and $V_2$ are independent random variables with $U_1\sim \Nmat(0,\alpha P_1)$, $U_2\sim \Nmat(0,\beta P_2)$, $V_1\sim\Nmat(0,\bar{\alpha}P_1)$ and $V_2 \sim \Nmat(0, \bar{\beta}P_2)$.
\end{proof}
In the following, we discuss capacity results for different interference regimes for MAZICs.

\section{MAZICs with Strong Interference}\label{sec: strong}
\subsection{Discrete Case}
Similar to \cite{Costa&ElGamal:87IT}, the discrete MAZIC with strong interference is defined as a discrete memoryless MAZIC satisfying
\bqa
I(X_1;Y_1|X_2)\leq I(X_1;Y_2|X_2X_3),\label{eq:strong_start}\\
I(X_2;Y_1|X_1)\leq I(X_2;Y_2|X_1X_3),\label{eq:strong_mid}\\
I(X_1X_2;Y_1)\leq I(X_1X_2;Y_2|X_3),\label{eq:strong_end}
\eqa
for all product distributions on $\Xmat_1\times\Xmat_2\times\Xmat_3$.

The above single letter conditions imply multi-letter conditions as stated below.
\begin{lemma}
For a discrete memoryless interference channel, if (\ref{eq:strong_start})-(\ref{eq:strong_end}) are satisfied for all product probability distributions on $\Xmat_1\times\Xmat_2\times\Xmat_3$, then
\bqa
I(X_1^n; Y_1^n|X_2^n)\leq I(X_1^n;Y_2^n|X_2^n X_3^n),\label{eq:multi_strong_1}\\
I(X_2^n; Y_1^n|X_1^n)\leq I(X_2^n;Y_2^n|X_1^nX_3^n),\label{eq:multi_strong_2}\\
I(X_1^nX_2^n;Y_1^n)\leq I(X_1^nX_2^n;Y_2^n|X_3^n).\label{eq:multi_strong_3}
\eqa
\end{lemma}
\vspace{10pt}
\begin{proof}
From the channel model, we have
\bqn
I(X_1^n; Y_1^n|X_2^nX_3^n)&=&I(X_1^n;Y_1^n|X_2^n),\\
I(X_2^n; Y_1^n|X_1^nX_3^n)&=&I(X_2^n;Y_1^n|X_1^n),\\
I(X_1^nX_2^n;Y_1^n|X_3^n)&=&I(X_1^nX_2^n; Y_1^n).
\eqn
The rest of the proof can be established using techniques similar to that of \cite{Costa&ElGamal:87IT}, hence is omitted.
\end{proof}
The above lemma leads to the following theorem.
\vspace{10pt}
\begin{theorem}\label{thm:discrete_strong}
For a discrete memoryless MAZIC with conditions (\ref{eq:strong_start})-(\ref{eq:strong_end}) for all product probability distributions on $\Xmat_1\times\Xmat_2\times\Xmat_3$, the capacity region is given by the set of all the nonnegative rate triples $(R_1,R_2,R_3)$ that satisfy
\bqa
R_1&\leq& I(X_1;Y_1|X_2Q),\label{eq:cr_strong_1}\\
R_2&\leq& I(X_2;Y_1|X_1Q),\label{eq:cr_strong_2}\\
R_3&\leq& I(X_3;Y_2|X_1X_2Q),\label{eq:cr_strong_3}\\
R_1+R_2&\leq& I(X_1X_2;Y_1|Q),\label{eq:cr_strong_4}\\
R_2+R_3&\leq& I(X_2X_3;Y_2|X_1Q),\label{eq:cr_strong_5}\\
R_1+R_3&\leq& I(X_1X_3;Y_2|X_2Q),\label{eq:cr_strong_6}\\
R_1+R_2+R_3 &\leq& I(X_1X_2X_3;Y_2|Q),\label{eq:cr_strong_7}
\eqa
where the input distribution factors as
\bqa
p(qx_1x_2x_3)=p(q)p(x_1|q)p(x_2|q)p(x_3|q).
\eqa
Furthermore, the region remains invariant if we impose the constraint $\|\Qmat\|\leq 8$.
\end{theorem}
\vspace{10pt}
The proof is given in Appendix \ref{apx:strong_cr}.
\subsection{Gaussian Case}
For a Gaussian MAZIC, the strong interference is defined as the case where $a\geq 1$ and $b\geq 1$, which are sufficient and necessary conditions for (\ref{eq:strong_start}) and (\ref{eq:strong_mid}), respectively. However, it is hard to find a sufficient and necessary conditions for (\ref{eq:strong_end}), and there are counter examples in which condition (\ref{eq:strong_end}) is violated even if $a\geq 1$ and $b \geq 1$. That is, there exist input distributions such that (\ref{eq:strong_end}) does not old with $a\geq 1$ and $b\geq 1$.

While Theorem \ref{cor:general_g_inner} still applies, a better rate splitting strategy can be devised for this case.
If $(R_1,R_2,R_3)$ is an achievable rate triple, then receiver $2$ can reliably recover $X_1$ and $X_2$ at these rates.
Therefore, receiver $2$ can decode whatever receiver $1$ decodes. Thus, if we choose the private message sets for users $1$
and $2$ to be empty, i.e., $\alpha=\beta=0$, we obtain an achievable rate region. 

In the following, we give an outer-bound on the capacity region.
\vspace{10pt}
\begin{corollary}\label{cor:CR_strong_gaussian}
For a Gaussian MAZIC with conditions $a,b\geq 1$, an outer-bound on the capacity region is given by the set of all the nonnegative rate triples $(R_1,R_2,R_3)$ that satisfy
\bqa
R_1 &\leq& \frac{1}{2}\log\left(1+P_1\right),\label{eq:strong_g_start}\\
R_2 &\leq& \frac{1}{2}\log\left(1+P_2\right),\label{eq:mac_r2}\\
R_3 &\leq& \frac{1}{2}\log\left(1+P_3\right),\label{eq:strong_cr_r3}\\
R_1+R_2 &\leq& \frac{1}{2}\log\left(1+P_1+P_2\right),\label{eq:mac_sum}\\
R_2+R_3 &\leq& \frac{1}{2}\log\left(1+bP_2+P_3\right),\label{eq:MAZIC_sum}\\
R_1+R_3 &\leq& \frac{1}{2}\log\left(1+aP_1+P_3\right).\label{eq:strong_g_end}
\eqa
\end{corollary}
\vspace{10pt}
The proof of this corollary is very similar to the proof of Theorem \ref{thm:discrete_strong}, except for the bound on $R_1+R_2+R_3$. The reason is that with $a\geq 1$ and $b\geq 1$, $I(X_1X_2;X_1+X_2+Z_1)\leq I(X_1X_2;\sqrt{a}X_1+\sqrt{b}X_2+Z_2)$ is generally not true for every possible input distribution, hence we do not have (\ref{eq:strong_end}). Therefore, inequality (\ref{eq:cr_strong_7}) cannot be obtained.

Next, let us consider one interference link being strong, for example, $1\leq a\leq 1+P_3$. In this case, we can easily get the following outer-bound:
\bqa
R_1&\leq&\frac{1}{2}\log(1+P_1),\label{eq:lineopt_ob_1}\\
R_2&\leq&\frac{1}{2}\log(1+P_2),\label{eq:lineopt_ob_2}\\
R_3&\leq&\frac{1}{2}\log(1+P_3),\label{eq:lineopt_ob_3}\\
R_1+R_2&\leq& \frac{1}{2}\log(1+P_1+P_2),\label{eq:lineopt_ob_4}\\
R_1+R_3&\leq& \frac{1}{2}\log(1+aP_1+P_3)\label{eq:lineopt_ob_5}.
\eqa

On the other hand, by setting $\alpha=\beta=0$ in the achievable region for Gaussian MAZICs in Corollary \ref{cor:general_g_inner},
one would have an achievable rate region with all nonnegative rate triples $(R_1,R_2,R_3)$ that satisfy
\bqa
R_1&\leq&\frac{1}{2}\log(1+P_1),\label{eq:onestrong_ib_1}\\
R_2&\leq&\frac{1}{2}\log(1+P_2),\label{eq:onestrong_ib_2}\\
R_3&\leq&\frac{1}{2}\log(1+P_3),\label{eq:onestrong_ib_3}\\
R_1+R_2&\leq&\frac{1}{2}\log(1+P_1+P_2),\label{eq:onestrong_ib_4}\\
R_1+R_3&\leq&\frac{1}{2}\log(1+aP_1+P_3),\label{eq:onestrong_ib_5}\\
R_2+R_3&\leq&\frac{1}{2}\log(1+bP_2+P_3),\label{eq:onestrong_ib_6}\\
R_1+R_2+R_3&\leq&\frac{1}{2}\log(1+aP_1+bP_2+P_3)\label{eq:onestrong_ib_7}.
\eqa

The following theorem summarizes the cases where some segment of the line: the intersection of the two hyperplanes defined by
\bqa
R_1+R_2 = \frac{1}{2}\log(1+P_1+P_2),\label{eq:line_1}\\
R_1+R_3 = \frac{1}{2}\log(1+aP_1+P_3)\label{eq:line_2}
\eqa
is on the boundary of the capacity region.
\vspace{5pt}
\begin{theorem}
For a Gaussian MAZIC with $1\leq a\leq 1+P_3$, if 
\bqa
b \geq\frac{1+aP_1+P_3}{1+P_1}\label{eq:lineopt_finalconditon},
\eqa
a segment of the line defined by (\ref{eq:line_1}) and (\ref{eq:line_2}), which starts at 
\begin{equation}\label{pt:endpoint1}
\left(\frac{1}{2}\log(1+P_1), \frac{1}{2}\log\left(1+\frac{P_2}{1+P_1}\right), \frac{1}{2}\log\left(1+\frac{P_3}{1+aP_1}\right)\right),
\end{equation} 
and ends at
\small
\begin{equation}\label{pt:endpoint2}
\left(\frac{1}{2}\log(1+P_1+P_2)-\frac{1}{2}\log\left(1+\frac{bP_2}{1+aP_1+P_3}\right), \frac{1}{2}\log\left(1+\frac{bP_2}{1+aP_1+P_3}\right), \frac{1}{2}\log\left(\frac{1+aP_1+bP_2+P_3}{1+P_1+P_2}\right)\right),
\end{equation}
\normalsize
is on the boundary of the capacity region of the channel.
\end{theorem}
\begin{proof}
Consider the rate triple $(R_1,R_2,R_3)$ on the line defined by (\ref{eq:line_1}) and (\ref{eq:line_2}). Any achievable rate triple on this line that also satisfies (\ref{eq:onestrong_ib_6}) and (\ref{eq:onestrong_ib_7}) must appear on the boundary of the capacity region as it belongs to both the inner and outer bounds.

Consider the rate triple defined by (\ref{pt:endpoint1}). It is achievable if
\begin{eqnarray}
\frac{1}{2}\log\left(1+\frac{P_2}{1+P_1}\right)\leq \frac{1}{2}\log\left(1+\frac{bP_2}{1+aP_1+P_3}\right),
\end{eqnarray}
i.e.,
\begin{equation}
b\geq \frac{1+aP_1+P_3}{1+P_1},
\end{equation}
as receiver $1$ first decodes $X_2$, subtracts it, and then decodes $X_1$; reciever $2$ also first decodes $X_2$, subtracts it, and then decodes $X_3$ by treating $X_1$ as noise.

The other rate triple defined by (\ref{pt:endpoint2}) satisfies (\ref{eq:onestrong_ib_7}) with equality, and satisfies (\ref{eq:onestrong_ib_6}) if $1\leq a\leq 1+P_3$ and $b\geq \frac{1+aP_1+P_3}{1+P_1}$.

Therefore, the line segment between these two rate triples (\ref{pt:endpoint1}) and (\ref{pt:endpoint2}) is on the boundary of the capacity region, and is achieved by time sharing.
\end{proof}

Fig.~\ref{fig:opt_line} gives an example where a line segment defined by (\ref{eq:line_1})  and (\ref{eq:line_2}) is on the boundary of the capacity region.
\begin{figure}[htp]
\centering
\scalefig{0.8}
\epsfbox{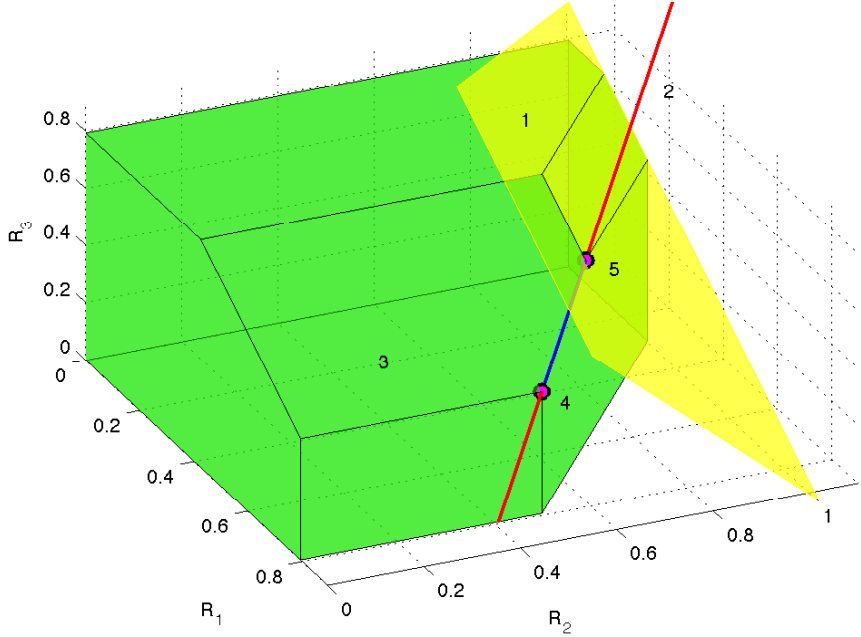}
\caption{The line $2$ defined in Eq. (\ref{eq:line_1}) and Eq. (\ref{eq:line_2}) appears as the boundary line of the capacity region. (Plane $1$ is defined by $R_1+R_2+R_3 = \frac{1}{2}\log(1+aP_1+bP_2+P_3)$; Region $3$ is defined by inequalities (\ref{eq:onestrong_ib_1})-(\ref{eq:onestrong_ib_7})); Points $4$ and $5$ are the two endpoints of the line segment that is on the capacity region. For this example, the corresponding channel parameters are: $a=1.2$, $b=3$, $P_1=P_3=2$, $P_2=3$.}\label{fig:opt_line} 
\end{figure}

Increasing $b$ even further for the case of $a\geq 1$ will ensure that (\ref{eq:onestrong_ib_6}) and (\ref{eq:onestrong_ib_7}) are never active. Specifically, we have
\begin{corollary}
For a Gaussian MAZIC with $a>1$ and $b>1+aP_1+P_3$, the capacity region is the set of all nonnegative rate triples $(R_1,R_2,R_3)$ that satisfies
\bqa
R_1&\leq&\frac{1}{2}\log(1+P_1),\\
R_2&\leq&\frac{1}{2}\log(1+P_2),\\
R_3&\leq&\frac{1}{2}\log(1+P_3),\\
R_1+R_2&\leq&\frac{1}{2}\log(1+P_1+P_2),\\
R_1+R_3&\leq&\frac{1}{2}\log(1+aP_1+P_3).
\eqa
\end{corollary}
\begin{proof}
With $a\geq 1$ and $b\geq 1+aP_1+P_3$, (\ref{eq:onestrong_ib_6}) and (\ref{eq:onestrong_ib_7}) are redundant in the achievable region. As a result, the inner-bound and outer-bound coincide with each other.
\end{proof}

\section{MAZICs with Very Strong Interference}\label{sec:verystrong}
\subsection{Discrete Case}
The discrete MAZIC with very strong interference is defined as a discrete memoryless MAZIC satisfying
\bqa
I(X_1;Y_1|X_2)&\leq& I(X_1;Y_2|X_2),\label{eq:verystrong_start}\\
I(X_2;Y_1|X_1)&\leq& I(X_2;Y_2|X_1),\label{eq:verystrong_mid}\\
I(X_1X_2;Y_1)&\leq& I(X_1X_2;Y_2).\label{eq:verystrong_end}
\eqa
for all product distributions on $\Xmat_1\times\Xmat_2\times\Xmat_3$.

It is easy to see that the condition specified by (\ref{eq:verystrong_start})-(\ref{eq:verystrong_end}) is a special case of the strong interference condition (\ref{eq:strong_start})-(\ref{eq:strong_end}). Therefore, one can immediately obtain the capacity region of the MAZIC with very strong interference from Theorem \ref{thm:discrete_strong}.
\vspace{10pt}
\begin{corollary}\label{cor:discrete_verystrong}
For a discrete memoryless MAZIC with conditions (\ref{eq:verystrong_start})-(\ref{eq:verystrong_end}) for all product probability distributions on $\Xmat_1\times\Xmat_2\times\Xmat_3$, the capacity region is given by the set of all the nonnegative rate triples $(R_1,R_2,R_3)$ that satisfy
\bqa
R_1&\leq& I(X_1;Y_1|X_2Q),\label{eq:cr_verystrong_1}\\
R_2&\leq& I(X_2;Y_1|X_1Q),\label{eq:cr_verystrong_2}\\
R_3&\leq& I(X_3;Y_2|X_1X_2Q),\label{eq:cr_verystrong_3}\\
R_1+R_2&\leq& I(X_1X_2;Y_1|Q),\label{eq:cr_verystrong_4}
\eqa
where the input distribution factors as
\bqa
p(qx_1x_2x_3)=p(q)p(x_1|q)p(x_2|q)p(x_3|q).
\eqa
Furthermore, the region remains invariant if we impose the constraint $\|\Qmat\|\leq 5$.
\end{corollary}
\vspace{10pt}

\subsection{Gaussian Case}
For a Gaussian MAZIC, very strong interference is defined as $a,b\geq 1+P_3$. Notice that the condition $a,b\geq 1+P_3$ is not a sufficient condition for (\ref{eq:verystrong_start}) and (\ref{eq:verystrong_mid}), as discussed in \cite[Theorem 2]{Xu-etal:Globecom10}. Again, it is a special case of the strong interference case, therefore, the capacity region can be readily obtained from Corollary \ref{cor:CR_strong_gaussian}.
\vspace{10pt}
\begin{corollary}
For a Gaussian MAZIC with conditions $a,b\geq 1+P_3$, the capacity region is given by the set of all nonnegative rate triples $(R_1,R_2,R_3)$ that satisfy
\bqa
R_1 &\leq& \frac{1}{2}\log\left(1+P_1\right),\label{eq:cr_verystrong_g_start}\\
R_2 &\leq& \frac{1}{2}\log\left(1+P_2\right),\\
R_3 &\leq& \frac{1}{2}\log\left(1+P_3\right),\\
R_1+R_2 &\leq& \frac{1}{2}\log\left(1+P_1+P_2\right)\label{eq:cr_verystrong_g_end}.
\eqa
\end{corollary}
\vspace{10pt}

\section{The MAZICs with mixed interference}\label{sec:mixed}
\subsection{Discrete Case}
The discrete MAZIC with mixed interference is defined as a discrete memoryless MAZIC satisfying
\bqa
p(y_1y_2|x_1x_2x_3) = p(y_1|x_1x_2)p(y_2|x_1x_2x_3) = p(y_1|x_1x_2)p'(y_2|x_3x_1y_1), \label{eq:mixed_weak}
\eqa
for some $p'(y_2|x_3x_1y_1)$, and
\bqa
I(X_2;Y_1|X_1)\leq I(X_2;Y_2|X_1X_3),\label{eq:mixed_strong}
\eqa
for all input distributions that factorizes as $p(x_1)p(x_2)p(x_3)^1$\footnote{$^1$(Condition \ref{eq:mixed_weak}) is refered to the link of weak interference, and condition (\ref{eq:mixed_strong}) is refered to the link of strong interference.}.

Condition (\ref{eq:mixed_weak}) means that we can find another discrete memoryless MAZIC with $(p(y_1|x_1x_2),p'(y_2|x_3x_1y_1))$ that has the same capacity region as the orginal MAZIC. Furthermore, the alternative MAZIC admits the Markov chain
\begin{equation}
X_1 -(X_2,X_3,Y_1)-Y_2.
\end{equation}

For this class of channel, we can outer-bound the capacity region as follows.

\begin{theorem}\label{thm:mixed_dm_ob}
For a discrete memoryless MAZIC with mixed interference, an outer-bound to the capacity region can be expressed as a set of nonnegative rate pairs $(R_1,R_2)$ satisfying the following inequalities:
\begin{eqnarray}
R_1&\leq& I(X_1;Y_1|X_2U_1Q),\label{eq:mixed_ob_1}\\
R_2&\leq& I(X_2;Y_1|X_1Q),\label{eq:mixed_ob_2}\\
R_3&\leq& I(X_3;Y_2|X_1X_2Q),\label{eq:mixed_ob_3}\\
R_3&\leq& I(U_1X_3;Y_1|Q),\label{eq:mixed_ob_4}\\
R_1+R_2 &\leq& I(X_1X_2;Y_1|Q), \label{eq:mixed_ob_5}\\
R_2+R_3 &\leq& I(X_2X_3;Y_2|X_1Q), \label{eq:mixed_ob_6}
\end{eqnarray}
where the input distribution is factorized as $p(q)p(u_1|q)p(x_1|u_1q)p(x_2|u_1q)p(x_3|q)$.
\end{theorem}

\begin{proof}
Inequalities (\ref{eq:mixed_ob_2}) and (\ref{eq:mixed_ob_3}) are trivial outer-bounds, and (\ref{eq:mixed_ob_5}) is the same as the sum-rate upper-bound for the MAC. Moreover, (\ref{eq:mixed_ob_6}) is the same as the sum-rate upper-bound for the two-user IC with strong interference \cite{Costa&ElGamal:87IT}. It remains to show (\ref{eq:mixed_ob_1}) and (\ref{eq:mixed_ob_4}).
First, let us consider
\begin{eqnarray}
n(R_1-\epsilon)&\stackrel{(a)}{\leq}& I(X_1^n;Y_1^n)\nonumber\\
&\stackrel{(b)}{\leq}&I(X_1^n;Y_1^n|X_2^n)\nonumber\\
&=&\sum_{i=1}^n I(X_1^n;Y_{1i}|X_2^nY_1^{i-1})\nonumber\\
&=&\sum_{i=1}^n \left\{H(Y_{1i}|X_2^nY_1^{i-1})-H(Y_{1i}|X_2^nY_1^{i-1}X_1^n)\right\}\nonumber
\end{eqnarray}
\begin{eqnarray}
&\stackrel{(c)}{=}&\sum_{i=1}^n \left\{H(Y_{1i}|X_2^{i-1}X_{2i}Y_1^{i-1})-H(Y_{1i}|X_{1i}X_{2i})\right\}\nonumber\\
&\stackrel{(d)}{\leq}&\sum_{i=1}^n \left\{H(Y_{1i}|X_{2i}U_{1i})-H(Y_{1i}|X_{1i}X_{2i}U_{1i})\right\}\nonumber\\
&=&\sum_{i=1}^n I(X_{1i};Y_{1i}|X_{2i}U_{1i}),\nonumber
\end{eqnarray}
where $(a)$ comes from Fano's inequality; $(b)$ is because of the independence between $X_1^n$ and $X_2^n$; $(c)$ is because that conditioning reduces entropy and the channel is assumed to be memoryless; for $(d)$, first we identify $U_{1i}=(X_2^{i-1}, Y_1^{i-1})$ and also the memoryless property induces the Markov chain $U_{1i}-(X_{1i}, X_{2i})-Y_{1i}$.

Now, let us show $X_{1i}-U_{1i}-X_{2i}$. Due to the memoryless property, the following Markov chain holds:
\begin{eqnarray}
(X_{1i}X_{2i})-(X_1^{i-1},X_2^{i-1})-Y_1^{i-1}.\nonumber
\end{eqnarray}
By weak union property, we obtain the following Markov chain:
\begin{eqnarray}
X_{2i}-(X_{1i},X_1^{i-1},X_2^{i-1})-Y_1^{i-1}.\nonumber
\end{eqnarray}
Together with the Markov chain $X_{2i}-X_2^{i-1}-X_{1i}X_1^{i-1}$, which due to the independence between $X_1^i$ and $X_2^i$, we obtain the following Markov chain by the contraction property:
\begin{eqnarray}
X_{2i}-X_2^{i-1}-(X_{1i},X_1^{i-1},Y_1^{i-1}).
\end{eqnarray}
Hence, we get the Markov chain
\begin{eqnarray}
X_{2i}-(X_2^{i-1},Y_1^{i-1})-X_{1i}
\end{eqnarray}
by the weak union and then the decomposition property.

Next, we consider
\begin{eqnarray}
n(R_3-\epsilon) &\stackrel{(a)}{\leq}& I(X_3^n;Y_2^n)\nonumber\\
&\stackrel{(b)}{\leq}& I(X_3^n;Y_2^n|X_2^n)\nonumber\\
&=& \sum_{i=1}^n I(X_3^n; Y_{2i}|X_2^nY_2^{i-1})\nonumber\\
&=&\sum_{i=1}^n \left\{H(Y_{2i}|X_2^nY_2^{i-1})-H(Y_{2i}|X_2^nX_3^nY_2^{i-1})\right\}\nonumber\\
&\stackrel{(c)}{\leq}&\sum_{i=1}^n \left\{H(Y_{2i}|X_{2i})-H(Y_{2i}|X_2^nX_3^nY_1^{i-1}Y_1^{i-1})\right\}\nonumber\\
&\stackrel{(d)}{=}&\sum_{i=1}^n \left\{H(Y_{2i}|X_{2i})-H(Y_{2i}|X_2^nX_3^nY_1^{i-1})\right\}\nonumber\\
&\stackrel{(e)}{=}&\sum_{i=1}^n \left\{H(Y_{2i}|X_{2i})-H(Y_{2i}|X_{2i}X_{3i}X_2^{i-1}Y_1^{i-1})\right\}\nonumber\\
&=&\sum_{i=1}^n \left\{I(X_{3i}U_{1i};Y_{2i}|X_{2i})\right\}\nonumber,
\end{eqnarray}
where $(a)$ follows the Fano's Inequality, $(b)$ is from the independence between $X_2^n$ and $X_3^n$; $(c)$ is because of the fact that conditioning reduces entropy; $(d)$ is due to the memoryless property of the channel, and the degradedness condition $X_1-(X_2,X_3,Y_1)-Y_2$, hence $Y_2^{i-1}$ is independent of any other random variables given $X_2^{i-1}$, $X_3^{i-1}$ and $Y_1^{i-1}$, then $(X_{2,i}^n, X_{3,i}^n, Y_{2i})-(X_2^{i-1},X_3^{i-1},Y_1^{i-1})-Y_2^{i-1}$ forms a Markov chain. By the weak union property, the Markov chain $Y_{2i}-(X_2^{n},X_3^n,Y_1^{i-1})-Y_2^{i-1}$ holds; $(e)$ is because of the Markov chain $(X_{2,i+1},X_3^{i-1},X_{3,i+1}^n)-(X_2^i,X_{3i},Y_1^{i-1})-Y_{2i}$. The easiest way to prove it is using the \textit{Independence Graph}. Alternatively, we first note that the Markov chain
\bqn
(X_2^{i-1},X_{2,i+1}^n,X_3^{i-1},X_{3,i+1}^n, Y_1^{i-1})-(X_{1i},X_{2i},X_{3i})-(Y_{1i},Y_{2i})
\eqn
holds because of the memoryless property of the channel.
By the decomposition property, the following Markov chain is obtained:
\bqn
(X_2^{i-1},X_{2,i+1}^n,X_3^{i-1},X_{3,i+1}^n, Y_1^{i-1})-(X_{1i},X_{2i},X_{3i})-Y_{2i}
\eqn
Further by the weak union property, we obtain the following Markov chain
\bqa
(X_{2,i+1}^n, X_3^{i-1}, X_{3,i+1}^n)-(X_{1i},X_2^i,X_{3i},Y_1^{i-1})-Y_{2i}.\label{mc1}
\eqa
On the other hand, again because of the memoryless property of the channel, the Markov chain
\bqn
(X_{1i},X_{2i},X_{3i},X_{2,i+1}^n,X_3^{i-1},X_{3,i+1}^n)-(X_1^{i-1},X_2^{i-1})-Y_1^{i-1}
\eqn
holds. Using the weak union property, we obtain the Markov chain
\bqn
(X_{2,i+1}^n,X_3^{i-1},X_{3,i+1}^n)-(X_{1i},X_{2i},X_{3i},X_1^{i-1},X_2^{i-1})-Y_1^{i-1}.
\eqn
Together with the markov chain
\bqn
(X_{2,i+1}^n,X_3^{i-1},X_{3,i+1}^n)-(X_2^{i-1}X_{2i}X_{3i})-(X_1^{i-1},X_{1i})
\eqn
due to the independence among $X_1^n$, $X_2^n$ and $X_3^n$,
we attain the Markov chain
\bqn
(X_{2,i+1}^n, X_3^{i-1},X_{3,i+1}^n)-(X_2^{i-1},X_{2i},X_{3i})-(X_1^{i-1},X_{1i},Y_1^{i-1})
\eqn
by the contraction property.
Then by the weak union property and the decomposition property, the Markov chain
\bqa
(X_{2,i+1}^n, X_3^{i-1},X_{3,i+1}^n)-(X_2^{i-1},X_{2i},X_{3i},Y_1^{i-1})-X_{1i}\label{mc2}
\eqa
holds.
Combine (\ref{mc1}) with (\ref{mc2}) by the contraction property, we have the Markov chain
\begin{equation}
(X_{2,i+1}^n, X_3^{i-1},X_{3,i+1}^n)-(X_2^{i-1},X_{2i},X_{3i},Y_1^{i-1})-(X_{1i},Y_{2i})\nonumber
\end{equation}
as desired. 
The rest of the proof is done by introducing the timesharing variable $Q$, similar to the proof of the capacity region for MACs \cite{Cover&Thomas:book}. 
\end{proof}

\subsection{Gaussian Case}
The mixed interference case corresponds to the condition $a \leq 1, b\geq 1$ or $a \geq 1, b\leq 1$ for the Gaussian MAZICs. As mentioned before, the notion of ``mixed'' differs from that of the classical two-user GIC with mixed interference: here the two interferences go to the same receiver. 

First of all, we can extend the outer-bound for the general discrete memoryless MAZICs to the Gaussian case.
\begin{corollary}\label{cor:mixed_g_ob}
For a Gaussian MAZIC with mixed interference ($a\leq 1$ and $b\geq 1$), an outer-bound to the capacity region can be expressed as a set of nonnegative rate pairs $(R_1,R_2)$ satisfying the following inequalities:
\begin{eqnarray}
R_1&\leq& \frac{1}{2}\log(1+\alpha P_1),\label{eq:mixed_g_ob_1}\\
R_2&\leq& \frac{1}{2}\log(1+P_2),\label{eq:mixed_g_ob_2}\\
R_3&\leq& \frac{1}{2}\log(1+P_3),\label{eq:mixed_g_ob_3}\\
R_3&\leq& \frac{1}{2}\log(1+\frac{a(1-\alpha)P_1+P_3}{1+a\alpha P_1}),\label{eq:mixed_g_ob_4}\\
R_1+R_2 &\leq& \frac{1}{2}\log(1+P_1+P_2), \label{eq:mixed_g_ob_5}\\
R_2+R_3 &\leq& \frac{1}{2}\log(1+bP_2+P_3), \label{eq:mixed_g_ob_6}
\end{eqnarray}
\end{corollary}

\begin{proof}
This is a direct extension of Theorem \ref{thm:mixed_dm_ob}. Inequalities (\ref{eq:mixed_g_ob_2}), (\ref{eq:mixed_g_ob_3}), (\ref{eq:mixed_g_ob_5}) and (\ref{eq:mixed_g_ob_6}) comes from the corresponding inequality in Theorem \ref{thm:mixed_dm_ob} and the fact that given the variance of random variables, Guassian distribution will maximize the entropy.

As for (\ref{eq:mixed_g_ob_4}),
\begin{eqnarray}
R_3 &\leq& I(UX_3;Y_2|X_2Q)\nonumber\\
&=&h(Y_2|X_2Q)-h(Y_2|X_2X_3UQ)\nonumber\\
&=&h(\sqrt{a}X_1+X_3+Z_2|Q)-h(\sqrt{a}X_1+Z_2|UQ)\nonumber\\
&\stackrel{(a)}{\leq}& \frac{1}{2}\log[(2\pi e)(1+aP_1+P_3)]-\frac{1}{2}\log a -h(X_1+Z_1+Z_2'|UQ)\nonumber\\
&\stackrel{(b)}{\leq}&\frac{1}{2}\log[(2\pi e)(1+aP_1+P_3)]-\frac{1}{2}\log a-\frac{1}{2}\log\left(2^{2h(X_1+Z_1|UQ)}+(2\pi e)(\frac{1-a}{a})\right)\nonumber\\
&\stackrel{(c)}{\leq}&\frac{1}{2}\log(1+aP_1+P_3)-\frac{1}{2}\log\left[a2^{2R_1}+1-a\right],\nonumber
\end{eqnarray}
where $(a)$ is by the fact that Gaussian distribution maximizes the entropy for a given variance, and $Z_2'\sim \mathcal{N}\left(0, \frac{1}{a}-1\right)$, independent of all other random variables; $(b)$ is from the entropy power inequality; $(c)$ is because that from (\ref{eq:mixed_ob_1}),
\begin{eqnarray}
R_1 \leq I(X_1;Y_1|X_2UQ) = h(Y_1|X_2UQ)-h(Z_1)=h(Y_1|X_2UQ)-\frac{1}{2}\log(2\pi e).\nonumber
\end{eqnarray}

Furthermore, since
\begin{eqnarray}
0\leq R_1\leq h(Y_1|X_2UQ)-h(Z_1)=h(X_1+Z_1|UQ)-h(Z_1)\leq h(X_1+Z_1|Q)-h(Z_1)\leq \frac{1}{2}\log(1+P_1),\nonumber
\end{eqnarray}
there exists an $\alpha \in [0,1]$, such that
\begin{eqnarray}
R_1 = \frac{1}{2}\log(1+\alpha P_1).
\end{eqnarray}
Then, 
\begin{eqnarray}
R_3\leq \frac{1}{2}\log (1+aP_1+P_3)-\frac{1}{2}\log(1+a\alpha P_1)=\frac{1}{2}\log\left(1+\frac{a(1-\alpha)P_1+P_3}{1+a\alpha P_1}\right).\nonumber
\end{eqnarray}
\end{proof}

\textit{Remark: } The outer-bound in Theorem \ref{thm:mixed_dm_ob} is an extension of Kramer's second outer-bound \cite[Thoerem 2]{Kramer:04IT} to the dicrete memoryless case. To see this, we can consider a special case of Corollary \ref{cor:mixed_g_ob} by choosing $R_2=0$, such that  the remaining transmitters $1$ and $3$, and receivers $1$ and $2$, form a  Gaussian ZIC. The outer bound in Corollary \ref{cor:mixed_g_ob} reduces to that consists of only (\ref{eq:mixed_g_ob_1}), (\ref{eq:mixed_g_ob_3}), and (\ref{eq:mixed_g_ob_4}) with the input distribution factorizes as $p(q)p(u|q)p(x_1|uq)p(x_3|q)$. If we lchoose $\beta = \frac{a\alpha P_1}{P}$, where $P=aP_1+P_3$, we can rewrite the outer bound  as:
\begin{eqnarray}
R_1&\leq& \frac{1}{2}\log(1+\frac{\beta P}{a}),\\
R_3&\leq& \frac{1}{2}\log(1+\frac{(1-\beta)P}{1+\beta P}),\nonumber
\end{eqnarray}
which is exactly Kramer's second outer bound on the capacity region  of a Gaussian ZIC \cite[Theorem 2]{Kramer:04IT}. Therefore, the  outer bound in Theorem \ref{thm:mixed_dm_ob} is a generalization of Kramer's outer bound to the discrete memoryless case, and an extension from the ZIC to the MAZIC.

In the following, we consider a subclass of Gaussian
MAZICs with mixed interference, and we determine some boundary points of the capacity region.
\vspace{10pt}
\begin{lemma}\label{lma:mixed_achievable}
For a Gaussian MAZIC satisfying conditions $a\leq 1$ and $b\geq 1+aP_1+P_3$, an achievable rate region is given by the set of all nonnegative rate triples $(R_1,R_2,R_3)$ that satisfy
\bqa
R_1\!\!\!\!&\leq&\!\!\!\!\frac{1}{2}\log\left(1+P_1\right),\label{eq:mixed_lma_start}\\
R_2\!\!\!\!&\leq&\!\!\!\! \frac{1}{2}\log\left(1+P_2\right),\\
R_3 \!\!\!\!&\leq&\!\!\!\! \frac{1}{2}\log\left(1+\frac{P_3}{1+a\alpha P_1}\right),\\
R_1+R_2 \!\!\!\!&\leq&\!\!\!\! \frac{1}{2}\log\left(1+P_1+P_2\right),\\
 R_1+R_3 \!\!\!\!&\leq&\!\!\!\! \frac{1}{2}\log\left(1+\alpha P_1\right)+\frac{1}{2}\log\left(1+\frac{a\bar{\alpha}P_1+P_3}{1+a\alpha P_1}\right),\\
R_1+R_2+R_3 \!\!\!\!&\leq&\!\!\!\! \frac{1}{2}\log\left(1+\alpha P_1+P_2\right)+\frac{1}{2}\log\left(1+\frac{a\bar{\alpha}P_1+P_3}{1+a\alpha P_1}\right),\label{eq:mixed_lma_end}
\eqa
for $\alpha\in[0,1]$.
\end{lemma}
\vspace{10pt}
\begin{proof}
If $b\geq 1+aP_1+P_3$, we know that receiver $2$ can decode user $2$'s message by treating its own signal as well as the interference from user
$1$ as noise. Therefore, there is no need to use rate splitting for user $2$, i.e., $\beta = 0$. On applying Corollary \ref{cor:general_g_inner} and
removing all the redundant inequalities, we get Lemma \ref{lma:mixed_achievable}.
\end{proof}

\textit{Remark: }$\frac{1}{2}\log\left(1+\alpha P_1+P_2\right)+\frac{1}{2}\log\left(1+\frac{a\bar{\alpha}P_1+P_3}{1+a\alpha P_1}\right)$ is
an increasing function of $\alpha$ if $a(1+P_2)\leq 1$. Thus, the maximal achievable sum rate for the above achievable rate region
is attained when $\alpha =1$, which equals $R_s=\frac{1}{2}\log(1+P_1+P_2)+\frac{1}{2}\log\left(1+\frac{P_3}{1+aP_1}\right)$.
However, since the expression of $R_s$ is generally not a concave function of $P_1$, we can achieve a larger sum rate than $R_s$ by time sharing. 

From Lemma \ref{lma:mixed_achievable} and Corollary \ref{cor:mixed_g_ob}, we can directly get a corner point on the capacity region.
\begin{corollary}
For a Gaussian MAZIC with $a\leq 1$ and $b\geq \frac{1+aP_1+P_3}{(1+P_1)}$, the rate triple $(R_1^\ast, R_2^\ast, R_3^\ast)$ is on the boundary of the capacity region, where
\bqa
R_1^\ast&=&\frac{1}{2}\log(1+P_1),\\
R_2^\ast&=&\frac{1}{2}\log\left(1+\frac{P_2}{1+P_1}\right),\\
R_3^\ast&=&\frac{1}{2}\log\left(1+\frac{P_3}{1+aP_1}\right).
\eqa
\end{corollary}
It is easy to see that this boundary point is achieved by fully decoding the interference from transmitter $2$ and treating
the interference from transmitter $1$ as noise.

\section{The MAZICs with weak interferences}\label{sec:weak}
\subsection{Discrete Memoryless Case}
\begin{definition}
A discrete memoryless MAZIC is said to have \textit{weak interferences} if the channel transition probability factorizes as
\bqa
p(y_1y_2|x_1x_2x_3)&=&p(y_1|x_1x_2)p'(y_2|x_2x_3y_1),\\
p(y_1y_2|x_1x_2x_3)&=&p(y_1|x_1x_2)p''(y_2|x_1x_3y_1)
\eqa
for some $p'(y_2|x_2x_3y_1)$ and $p''(y_2|x_1x_3y_1)$, or, equivalently, the channel is stochastically degraded.
\end{definition}

In the absence of receiver cooperation, a stochastically degraded interference channel is equivalent in its capacity to a physically degraded interference channel. As such, we will assume in the following that the channel is physically degraded, i.e., the MAZIC admits the Markov chains $X_1-(X_2, X_3,Y_1)-Y_2$ and $X_2-(X_1, X_3, Y_1)-Y_2$. As a consequence, the following two inequalities hold
\begin{eqnarray}
I(U_1;Y_2|X_2X_3)&\leq& I(U_1;Y_1|X_2),\label{ieq:weak1}\\
I(U_2;Y_2|X_1X_3)&\leq& I(U_2;Y_1|X_1)\label{ineq:weak2}
\end{eqnarray}
for all input distributions $p(x_3)p(u_1)p(x_1|u_1)p(x_2|u_1)$ and $p(x_3)p(u_2)p(x_1|u_2)p(x_2|u_2)$ respectively.

The above definition of weak interference leads to the following outer-bound.
\begin{theorem}\label{thm:weak_ob}
The capacity region of a discrete memoryless MAZIC with weak interferences is outer-bounded by the region determined by the following inequalites:
\begin{eqnarray}
R_1&\leq& I(X_1;Y_1|X_2U_1Q),\label{ineq:weak_ob_1}\\
R_2&\leq& I(X_2;Y_1|X_1U_2Q),\label{ineq:weak_ob_2}\\
R_3&\leq& I(X_3;Y_2|X_1X_2Q),\label{ineq:weak_ob_3}\\
R_3&\leq& I(X_3U_1;Y_2|X_2Q),\label{ineq:weak_ob_4}\\
R_3&\leq& I(X_3U_2;Y_2|X_1Q),\label{ineq:weak_ob_5}\\
R_1+R_2&\leq& I(X_1X_2;Y_1|Q),\label{ineq:weak_ob_6}
\end{eqnarray}
where the input distribution $p(u_1u_2x_1x_2x_3)=p(u_1u_2)p(x_1|u_1u_2)p(x_2|u_1u_2)p(x_3)$. 
\end{theorem}
The proof is similar to that of Theorem 4  and is hence omitted. We note that the auxiliary random variables are defined as $U_{1i}=(X_2^{i-1}, Y_1^{i-1})$ and $U_{2i}=(X_1^{i-1}, Y_1^{i-1})$.

\subsection{Gaussian Case}
The weak interference case for the Gaussian MAZIC corresponds to the condition with $a, b\leq 1$.

First, Theorem \ref{thm:weak_ob} can be extended to the Gaussian case.
\begin{corollary}\label{cor:weak_ob}
For a Gaussian MAZIC satisfying conditions $a,b\leq 1$, an outer bound to the capacity region is given by the
set of all nonnegative rate triples $(R_1,R_2,R_3)$ such that
\bqn
R_1\!\!\!\!&\leq&\!\!\!\!\frac{1}{2}\log(1+\alpha P_1),\label{eq:ob_weak_1}\\
R_2\!\!\!\!&\leq&\!\!\!\!\frac{1}{2}\log(1+\beta P_2),\label{eq:ob_weak_2}\\
R_3\!\!\!\!&\leq&\!\!\!\!\frac{1}{2}\log(1+P_3),\label{eq:ob_weak_3}\\
R_3\!\!\!\!&\leq&\!\!\!\!\frac{1}{2}\log\left(1+\frac{a(1-\alpha)P_1+P_3}{1+a\alpha P_1}\right),\label{eq:ob_weak_4}\\
R_3\!\!\!\!&\leq&\!\!\!\!\frac{1}{2}\log\left(1+\frac{b(1-\beta)P_2+P_3}{1+b\beta P_2}\right),\label{eq:ob_weak_5}\\
R_1+R_2\!\!\!\!&\leq&\!\!\!\!\frac{1}{2}\log(1+P_1+P_2).\label{eq:ob_weak_6}
\eqn
\end{corollary}

The proof is very similar to that of Corollary \ref{cor:mixed_g_ob}, hence is omitted here.

For a two-user Gaussian ZIC, treating interference as noise is optimal in terms of sum-capacity for the weak interference case. One may conjecture that a similar result holds for the Gaussian MAZIC if both interferences are weak ($a,b\leq 1$). Indeed, similar sum-rate capacity result holds for the case with $0<a=b<1$. 
\begin{corollary}
For the Gaussian MAZICs satisfying $0\leq a=b\leq 1$, the sum-rate capacity is 
\bqa
C=\frac{1}{2}\log(1+P_1+P_2)+\frac{1}{2}\log\left(1+\frac{P_3}{1+aP_1+bP_2}\right).
\eqa
\end{corollary}
\begin{proof}
This is a direct extension of the sum-capacity result of the two-user Gaussian ZICs with weak interference by viewing $X_1$ and $X_2$ as a group.
\end{proof}

However, the above sum-capacity result is not true in general with asymmetric interference. We begin with the following theorem that gives a sum-rate upper-bound.

\begin{theorem}\label{thm:weak_sr_ob}
Any achievable rate triplet ($R_1$, $R_2$, $R_3$) for the Gaussian MAZIC with $0\leq a\leq b\leq 1$ must satisfy the following constraint
\bqn
n(R_1+R_2+R_3)
&\leq& \min_{\sigma^2\leq 1}\left\{\frac{n}{2}\log\left((P_1+P_2+1)(aP_1+bP_2+\sigma^2)-(\sqrt{a}P_1+\sqrt{b}P_2+\sqrt{a})^2\right)\right.\\&&\left.-\frac{n}{2}\log(aP_1+bP_2+1)-\frac{n}{2}\log(\sigma^2-a)+\frac{n}{2}\log(aP_1+bP_2+P_3+1)\right\}.
\eqn
\end{theorem}
\begin{proof}
\bqn
&&n(R_1+R_2+R_3)-n\epsilon\\
&\stackrel{(a)}{\leq}&I(X_1^nX_2^n; Y_1^n)+I(X_3^n;Y_2^n)\\
&=&I(X_1^n;X_1^n+Z_1^n)+I(X_2^n;X_1^n+X_2^n+Z_1^n)+I(X_3^n;\sqrt{a}X_1^n+\sqrt{b}X_2^n+X_3^n+Z_2^n)\\
&\stackrel{(b)}{\leq}&I(X_1^n;X_1^n+Z_1^n)+I(X_2^n;X_1^n+X_2^n+Z_1^n,\sqrt{a}X_1^n+\sqrt{b}X_2^n+N_1^n)\\&&+I(X_3^n;\sqrt{a}X_1^n+\sqrt{b}X_2^n+X_3^n+Z_2^n)\\
&=&h(X_1^n+Z_1^n)-h(Z_1^n)+h(\sqrt{a}X_1^n+\sqrt{b}X_2^n+N_1^n)\\&&+h(X_1^n+X_2^n+Z_1^n|\sqrt{a}X_1^n+\sqrt{b}X_2^n+N_1^n)-h(\sqrt{a}X_1^n+N_1^n)-h(X_1^n+Z_1^n|\sqrt{a}X_1^n+N_1^n)\\&&+h(\sqrt{a}X_1^n+\sqrt{b}X_2^n+X_3^n+Z_2^n)-h(\sqrt{a}X_1^n+\sqrt{b}X_2^n+Z_2^n)\\
&=&h(X_1^n+Z_1^n)-h(\sqrt{a}X_1^n+N_1^n)+h(\sqrt{a}X_1^n+\sqrt{b}X_2^n+N_1^n)\\
&&-h(\sqrt{a}X_1^n+\sqrt{b}X_2^n+Z_2^n)-h(Z_1^n)+h(X_1^n+X_2^n+Z_1^n|\sqrt{a}X_1^n+\sqrt{b}X_2^n+N_1^n)\\
&&+h(\sqrt{a}X_1^n+\sqrt{b}X_2^n+X_3^n+Z_2^n)-h(Z_1^n-\frac{1}{\sqrt{a}}N_1^n|\sqrt{a}X_1^n+N_1^n)\\
&=&h(X_1^n+Z_1^n)-h(\sqrt{a}X_1^n+N_1^n|Z_1^n-\frac{1}{\sqrt{a}}N_1^n)+h(\sqrt{a}X_1^n+\sqrt{b}X_2^n+N_1^n)\\
&&-h(\sqrt{a}X_1^n+\sqrt{b}X_2^n+Z_2^n)-h(Z_1^n)+h(X_1^n+X_2^n+Z_1^n|\sqrt{a}X_1^n+\sqrt{b}X_2^n+N_1^n)\\
&&+h(\sqrt{a}X_1^n+\sqrt{b}X_2^n+X_3^n+Z_2^n)-h(Z_1^n-\frac{1}{\sqrt{a}}N_1^n)\\
&\stackrel{(c)}{\leq}& \frac{n}{2}\log\left((P_1+P_2+1)(aP_1+bP_2+\sigma^2)-(\sqrt{a}P_1+\sqrt{b}P_2+\sqrt{a})^2\right)\\&&-\frac{n}{2}\log(aP_1+bP_2+1)-\frac{n}{2}\log(\sigma^2-a)+\frac{n}{2}\log(aP_1+bP_2+P_3+1)
\eqn
where $(a)$ is from Fano's inequality; $(b)$ is by giving side information $\sqrt{a}X_1^n+\sqrt{b}X_2^n+N_1^n$ to the second mutual information where $N_1^n$ is an i.i.d. Gaussian random variables whose covariance matrix with $Z_1$ is 
\bqn
Cov\left[\begin{array}{c}Z_1\\N_1\end{array}\right]=\left[\begin{array}{cc} 1 & \rho\sigma\\ \rho\sigma & \sigma^2\end{array}\right];
\eqn
$(c)$ is the result of applying the extremal inequality \cite{Liu&Viswanath:06IT} to the first two terms, and to the third and forth terms respectively. for the first two terms,
\bqn
h(X_1^n+Z_1^n)-h(\sqrt{a}X_1^n+N_1^n|Z_1^n-\frac{1}{\sqrt{a}}N_1^n)&\leq& \frac{n}{2}\log(1+P_1)-\frac{n}{2}\log(aP_1+a)\\
&=&-\frac{n}{2}\log a,
\eqn
since the use of the extremal inequality requires $Var(N_1|Z_1-\frac{1}{\sqrt{a}}N_1)\geq a\Rightarrow \rho\sigma=\sqrt{a}$. For the third and fourth terms,
\bqn
h(\sqrt{a}X_1^n+\sqrt{b}X_2^n+N_1^n)-h(\sqrt{a}X_1^n+\sqrt{b}X_2^n+Z_2^n)&\leq& \frac{n}{2}\log(aP_1+bP_2+\sigma^2)-\frac{n}{2}\log(aP_1+bP_1+1)
\eqn
as the use of the extremal inequality requires $\sigma^2\leq 1$.

For the conditional entropy $h(X_1^n+X_2^n+Z_1^n|\sqrt{a}X_1^n+\sqrt{b}X_2^n+N_1^n)$, identically and independently distributed (i.i.d) zero-mean Gaussian $X_1^n$ and $X_2^n$ are the maximizing distributions \cite{Thomas:87IT}.
\end{proof}

\begin{corollary}
For the Gaussian MAZICs satisfying $0\leq a\leq b\leq 1$, if the power constraints satisfy
\bqn
P_1&=&\frac{1-\sqrt{ab}}{\sqrt{ab}-a},\\
P_3&\geq& b-1+(b-a)P_1 = \sqrt{\frac{b}{a}}-\sqrt{ab},
\eqn
the sum-rate capacity is 
\bqa
C=\frac{1}{2}\log(1+P_1)+\frac{1}{2}\log\left(1+\frac{bP_2+P_3}{1+aP_1}\right). \label{eq:sc_weak}
\eqa
\end{corollary}
\begin{proof}
For the achievability part, let receiver $1$ decode messages from users $1$ and $2$, and receiver $2$ decode messages from users $2$ and $3$, we have the following achievable rate triplets $(R_1,R_2,R_3)$:
\begin{eqnarray}
R_1&\leq& \frac{1}{2}\log(1+P_1),\\
R_2&\leq& \frac{1}{2}\log\left(1+\frac{bP_2}{1+aP_1}\right),\\
R_3&\leq& \frac{1}{2}\log\left(1+\frac{P_3}{1+aP_1}\right),\\
R_1+R_2&\leq&\frac{1}{2}\log(1+P_1+P_2),\\
R_2+R_3&\leq&\frac{1}{2}\log\left(1+\frac{bP_2+P_3}{1+aP_1}\right).
\end{eqnarray} 
Apply Fourier-Motzkin elimination with respect to $S = R_1+R_2+R_3$, the resulting achievable sum-rate is 
\bqn
R_1+R_2+R_3 \leq \min\left\{\frac{1}{2}\log(1+P_1)+\frac{1}{2}\log\left(1+\frac{bP_2+P_3}{1+aP_1}\right), \frac{1}{2}\log(1+P_1+P_2)+\frac{1}{2}\log\left(1+\frac{P_3}{1+aP_1}\right)\right\},
\eqn
if $(b-a)P_1\leq 1-b+P_3$, 
\bqn
\frac{1}{2}\log(1+P_1+P_2)+\frac{1}{2}\log\left(1+\frac{P_3}{1+aP_1}\right)\geq  \frac{1}{2}\log(1+P_1)+\frac{1}{2}\log\left(1+\frac{bP_2+P_3}{1+aP_1}\right).
\eqn
hence, $\frac{1}{2}\log(1+P_1)+\frac{1}{2}\log\left(1+\frac{bP_2+P_3}{1+aP_1}\right)$ is an achievable sum-rate, and is achieved by user $1$ decoding $X_2$ first, subtracting it off, and then decoding $X_1$; and user $2$ decoding $X_2$ and $X_3$ simultaneously by treating $X_1$ as noise.

For the converse part, at the last step of the proof of Theorem \ref{thm:weak_sr_ob}, if we further let the Gaussian variables $X_2^n-(\sqrt{a}X_1^n+\sqrt{b}X_2^n+N_1^n)- (X_1^n+X_2^n+Z_1^n)$ form a Markov chain,
then 
\bqa
P_1 = \frac{\sqrt{ab}-\sigma^2}{a-\sqrt{ab}}.\label{eq:P1_constraint}
\eqa
The sum-rate upper-bound becomes
\bqn
\frac{1}{2}\log(1+P_1)+\frac{1}{2}\log(1+\frac{bP_2}{aP_1+\sigma^2})+\frac{1}{2}\log(1+\frac{P_3}{1+aP_1+bP_2}).
\eqn
Let $\sigma^2=1$, (\ref{eq:P1_constraint}) becomes $P_1 = \frac{1-\sqrt{ab}}{\sqrt{ab}-a}$, naturally, this requires $a\leq b$, and $\sqrt{ab}\leq 1$ such that (\ref{eq:P1_constraint}) is non-negative. This is because $a>b$ is infeasible as it implies $\sqrt{ab}\leq a$, i.e., (\ref{eq:P1_constraint}) is negative when $\sigma^2=1$.
\end{proof}

It is perhaps not intuitive that the sum-rate (\ref{eq:sc_weak}) is optimal only if $P_1 = \frac{1-\sqrt{ab}}{\sqrt{ab}-a}$. Specifically, given that this sum-rate capacity is achieved when the interference from $X_1$ is treated as noise at $Y_2$, it might be expected that with smaller $P_1$, the same scheme should also be optimal. We show that this is not true.

First, for $a\leq 1$,
\begin{equation}
\frac{1-b}{b-a}\leq \frac{1-\sqrt{ab}}{\sqrt{ab}-a}.\nonumber
\end{equation}
But for $P_1\leq \frac{1-b}{b-a}$, the achievable sum-rate
\begin{equation}
\frac{1}{2}\log(1+P_1+P_2)+\frac{1}{2}\log\left(1+\frac{P_3}{1+aP_1+bP_2}\right) \label{eq:weak_tian}
\end{equation} 
is greater than the sum-rate (\ref{eq:sc_weak}). 

Now consider any $P_1$ with $\frac{1-b}{b-a}\leq P_1\leq \frac{1-\sqrt{ab}}{\sqrt{ab}-a}$. The following function is an achievable sum-rate for $P_1\leq \frac{1-\sqrt{ab}}{\sqrt{ab}-a}$. However, it is easy to show that $f$ is not concave in $P_1$ around the point $\frac{1-b}{b-a}$. Therefore, sum-rates strictly larger than (\ref{eq:sc_weak}) can be achieved for $\frac{1-b}{b-a}\leq P_1\leq \frac{1-\sqrt{ab}}{\sqrt{ab}-a}$ using time-sharing.

\bqn
f(P_1)=\left\{\begin{array}{ll}
\frac{1}{2}\log(1+P_1+P_2)+\frac{1}{2}\log\left(1+\frac{P_3}{1+a P_1+b P_2}\right), & \textrm{if    } P_1\leq \frac{1-b}{b-a},\\
\frac{1}{2}\log(1+ P_1)+\frac{1}{2}\log\left(1+\frac{bP_2+P_3}{1+a P_1}\right), & \textrm{if    } \frac{1-b}{b-a}\leq P_1 \leq \frac{1-\sqrt{ab}}{\sqrt{ab}-a}.
\end{array}\right.
\eqn

\begin{figure}[htp]
\centering
\scalefig{1.0}
\epsfbox{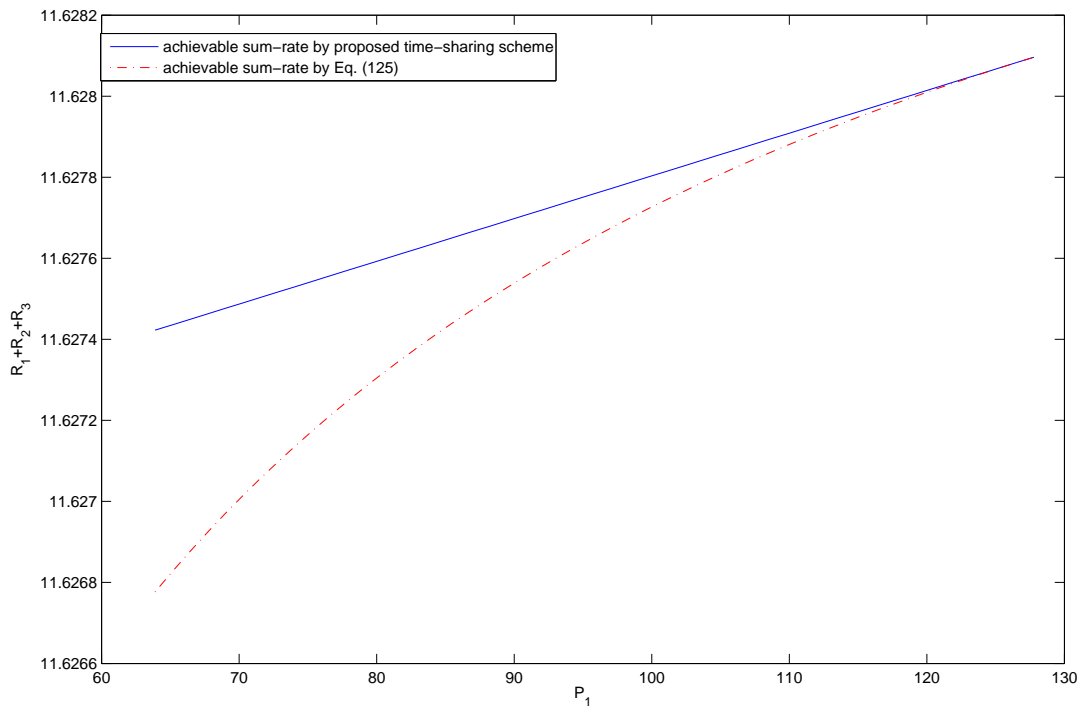}
\caption{The Comparison of the sum-rates achieved by proposed time-sharing scheme and Eq. \ref{eq:sc_weak} when $\frac{1-b}{b-a}\leq P_1\leq \frac{1-\sqrt{ab}}{\sqrt{ab}-a}$.}\label{fig:sum_rate_weak}
\end{figure}
Next, let us consider an even simpler case, where one of the cross link gain vanishes, for example, $a=0$. With only one weak interference link, we are able to obtain a boundary curve of the capacity region.

\begin{theorem}
For a Gaussian MAZIC with $a=0$ and $\frac{1+P_3}{1+P_1}\leq b\leq 1$ ($P_3\leq P_1$), then the following rate triple is always on the boundary of the capacity region:
\begin{eqnarray}
\left(\frac{1}{2}\log\left(1+\frac{P_1}{1+\bar{\beta}P_2}\right), \frac{1}{2}\log(1+\bar{\beta} P_2)+\frac{1}{2}\log\left(1+\frac{\beta P_2}{1+P_1+\bar{\beta}P_2}\right), \frac{1}{2}\log(1+P_3)\right),\label{rate}
\end{eqnarray}
where $\beta\in [0,1]$ and satisfy
\begin{eqnarray}
\frac{1}{2}\log(1+\bar{\beta} P_2)+\frac{1}{2}\log\left(1+\frac{\beta P_2}{1+P_1+\bar{\beta}P_2}\right)\leq \frac{1}{2}\log\left(1+\frac{bP_2}{1+P_3}\right).
\end{eqnarray}
\end{theorem}
\begin{proof}
By setting $\alpha=1$, the general achievable rate region in Corollary \ref{cor:general_g_inner} reduces to
\begin{eqnarray}
R_1&\leq&\frac{1}{2}\log(1+P_1),\nonumber\\
R_2&\leq&\frac{1}{2}\log(1+P_2),\nonumber\\
R_3&\leq&\frac{1}{2}\log\left(1+\frac{P_3}{1+b\beta P_2}\right),\nonumber\\
R_1+R_2&\leq&\frac{1}{2}\log(1+P_1+P_2),\nonumber\\
R_2+R_3&\leq&\frac{1}{2}\log(1+\beta P_2)+\frac{1}{2}\log\left(1+\frac{b\bar{\beta}P_2+P_3}{1+b\beta P_2}\right),\nonumber\\
R_1+R_2+R_3&\leq&\frac{1}{2}\log(1+P_1+\beta P_2)+\frac{1}{2}\log\left(1+\frac{b\bar{\beta}P_2+P_3}{1+b\beta P_2}\right).\nonumber
\end{eqnarray}
If let $R_3=\frac{1}{2}\log(1+P_3)$, the achievable rate region reduces to
\begin{eqnarray}
R_1&\leq& \frac{1}{2}\log(1+P_1),\\
R_2&\leq&\frac{1}{2}\log\left(1+\frac{bP_2}{1+P_3}\right),\\
R_1+R_2&\leq&\frac{1}{2}\log(1+P_1+P_2).\label{eq:thm7_sumrate}
\end{eqnarray}
If $b\geq \frac{1+P_3}{1+P_1}$ ($P_3\leq P_1$), inequality (\ref{eq:thm7_sumrate}) is always active. Therefore, the rate triple (\ref{rate}) is always achievable.

For the converse part, (\ref{eq:thm7_sumrate}) is a natural upper-bound for $R_1+R_2$.
\end{proof}

\section{Conclusion}\label{sec:conclusion}
In this paper we have studied the capacity of an uplink network with co-channel interference. By modeling such networks using
a multiple access interference channel with one-sided interference,
we have obtained an inner bound to the capacity region for both the discrete memoryless case and the Gaussian case.
The capacity region for the discrete memoryless channel model with strong and very strong interference has been established; for the Gaussian MAZIC, we have determined the capacity region for the very strong interference case, and for the case that one interference link being strong and the other one being very strong; for the strong interference case, we have obtained a boundary line segment of the capacity region. For the mixed interference case, a boundary point of the capacity region has been obtained. For the weak interference case, we have obtained the sum-rate capacity for the symmetric channel coefficients whose result is analogous to that of the two user Gaussian one-sided interference channel. For the general case, a sum-rate upper bound has been obtained which gives rise to a sum-rate capacity result under certain power constraint conditions. Furthermore, it does not change the capacity results if we allow more users intended for receiver $2$ without interfering receiver $1$. In this case, $R_3$ is replaced by the sum-rate of all those added users.

\begin{appendix}
\subsection{Proof of Theorem \ref{thm:general_inner}}\label{apx:general_inner}
Fix $p(q)p(u_1|q)p(x_1|u_1q)p(u_2|q)p(x_2|u_2q)p(x_3|q)$.

Codebook generation: Randomly generate a time sharing sequence $q^n$ according to $\prod_{i=1}^np(q_i)$. Randomly generate $2^{nR_3}$ sequences $x_3^n(m_3)$, $m_3\in [1:2^{nR_3}]$, according to $\prod_{i=1}^n p(x_{3i}|q_i)$. For $j=1,2$, randomly generate $2^{nT_i}$ sequences $u_j^n(l_j)$, $l_j\in [1:2^{nT_j}]$, each according to $\prod_{i=1}^n p_{U_j|Q}(u_{ji}|q_i)$. For each $u_j^n(l_j)$, randomly generate $2^{nS_j}$ sequences $x_j^n(l_j,k_j)$, $k_j\in [1:2^{nS_j}]$, each according to $\prod_{i=1}^n p_{X_j|U_j,Q}(x_j|u_{ji}(l_j),q_i)$. The codebook is available at all transmitters and receivers.

Encoding: For user $j$, $j=1,2$, to send message $m_j=(l_j,k_j)$, encoder $j$ transmits $x_j^n(l_j,k_j)$. For user $3$, to send message $m_3$, encoder $3$ transmits $x_3^n(m_j)$.

Decoding: Upon receiving $y_1^n$, decoder $1$ finds the unique message tuple $(\hat{l}_1,\hat{l}_2,\hat{k}_1,\hat{k}_2)$ such that 
\begin{eqnarray}
(q^n,u_1^n(\hat{l}_1),u_2^n(\hat{l}_2),x_1^n(\hat{l}_1,\hat{k}_1),x_2^n(\hat{l}_2,\hat{k}_2),y_1^n)\in A_\epsilon^{(n)}(QU_1U_2X_1X_2Y_1).
\end{eqnarray}
If no such unique tuple exists, the decoder declares an error.

Upon receiving $y_2^n$, decoder $2$ finds the unique message $\hat{m}_3$ such that
\begin{equation}
(q^n, u_1^n(l_1),u_2^n(l_2), x_3^n(\hat{m}_3))\in A_\epsilon^{(n)}(QU_1U_2X_3Y_2),
\end{equation}
for some $l_1\in [1:2^{nT_1}]$ and some $l_2\in[1:2^{nT_2}]$.
If no such unique $\hat{m}_3$ exists, the decoder declares an error.

Analysis of the probability of error: By the symmetry of the codebook generation, we assume that the transmitted indices are $l_1=l_2=k_1=k_2=m_3=1$. For user $1$, we
define the following event:
\bqa
E_{l_1l_2k_1k_2}^1=\left\{(q^n,u_1^n(l_1),u_2^n(l_2),x_1^n(l_1,k_1),x_2^n(l_2,k_2),y_1^n)\in A_\epsilon^{(n)}(QU_1U_2X_1X_2Y_1)\right\}.
\eqa
The error probability at receiver $1$ is

\bqn
P_{e1}^n&=&\Pr\left\{{E_{1111}^1}^c\bigcup\cup_{(l_1l_2k_1k_2)\neq(1,1,1,1)}E_{l_1l_2k_1k_2}^1\right\}\\
&\leq&\Pr({E_{1111}^1}^c)+\sum_{l_1\neq 1,l_2=k_1=k_2=1} \Pr(E_{l_1111}^1)+\sum_{l_2\neq 1,l_1=k_1=k_2=1}\Pr(E_{1l_211}^1)+\sum_{k_1\neq 1,l_1=l_2=k_2=1}\Pr(E_{11k_11}^1)\\
&&\sum_{k_2\neq 1,l_1=l_2=k_1=1}\Pr(E_{111k_2}^1)+\sum_{l_1,l_2\neq 1,k_1=k_2=1}\Pr(E_{l_1l_211}^1)+\sum_{l_1,k_1\neq 1,l_2=k_2=1}\Pr(E_{l_11k_11}^1)\\&&\sum_{l_1,k_2\neq 1,l_2=k_1=1}\Pr(E_{l_111k_2}^1)+\sum_{l_2,k_1\neq 1,l_1=k_2=1}\Pr(E_{1l_2k_11}^1)+\sum_{l_2,k_2\neq 1,l_1=k_1=1}\Pr(E_{1l_21k_2}^1)\\&&\sum_{k_1,k_2\neq 1,l_1=l_2=1}\Pr(E_{11k_1k_2}^1)+\sum_{l_1,l_2,k_1\neq 1,k_2=1}\Pr(E_{l_1l_2k_11}^1)+\sum_{l_1,l_2,k_2\neq 1,k_1=1}\Pr(E_{l_1l_21k_2}^1)\\&&\sum_{l_1,k_1,k_2\neq 1,l_2=1}\Pr(E_{l_11k_1k_2}^1)+\sum_{l_2,k_1,k_2\neq 1,l_1=1}\Pr(E_{1l_1k_1k_2}^1)+\sum_{l_1,l_2,k_1,k_2\neq 1}\Pr(E_{l_1l_2k_1k_2}^1)
\eqn

It is obvious that $\Pr({E_{1111}^1}^c)\rightarrow 0$ when $n\rightarrow \infty$. From the joint typicality we have
\bqn
&&\sum_{l_1\neq 1,l_2=k_1=k_2=1}\Pr(E_{l_1111}^1)\\
&\leq& 2^{nT_1}\sum_{(q^n,u_1^n,u_2^n,x_1^n,x_2^n,y_1^n)\in A_\epsilon^{(n)}}p(u_1^n,x_1^n|q^n)p(q^nu_2^nx_2^ny_1^n)\\
&\leq& 2^{nT_1}2^{n(H(QU_1U_2X_1X_2Y_1)+\epsilon)}2^{-n(H(U_1X_1|Q)-2\epsilon)}2^{-n(H(QU_2X_2Y_1)-\epsilon)}\\
&=&2^{n(T_1-I(U_1X_1;Y_1|U_2X_2Q)+4\epsilon)}=2^{n(T_1-I(X_1;Y_1|X_2Q)+4\epsilon)}
\eqn
\bqn
&&\sum_{l_2\neq 1,l_1=k_1=k_2=1}\Pr(E_{1l_211}^1)\\
&\leq& 2^{nT_2}\sum_{(q^n,u_1^n,u_2^n,x_1^n,x_2^n,y_1^n)\in A_\epsilon^{(n)}}p(u_2^n,x_2^n|q^n)p(q^nu_1^nx_1^ny_1^n)\\
&\leq& 2^{nT_2}2^{n(H(QU_1U_2X_1X_2Y_1)+\epsilon)}2^{-n(H(U_2X_2|Q)-2\epsilon)}2^{-n(H(QU_1X_1Y_1)-\epsilon)}\\
&=&2^{n(T_2-I(U_2X_2;Y_1|U_1X_1Q)+4\epsilon)}=2^{n(T_2-I(X_2;Y_1|X_1Q)+4\epsilon)}
\eqn
\bqn
&&\sum_{k_1\neq 1,l_1=l_2=k_2=1}\Pr(E_{11k_11}^1)\\
&\leq& 2^{nS_1}\sum_{(q^n,u_1^n,u_2^n,x_1^n,x_2^n,y_1^n)\in A_\epsilon^{(n)}}p(x_1^n|u_1^n,q^n)p(q^nu_1^nu_2^nx_2^ny_1^n)\\
&\leq& 2^{nS_1}2^{n(H(QU_1U_2X_1X_2Y_1)+\epsilon)}2^{-n(H(X_1|U_1Q)-2\epsilon)}2^{-n(H(QU_1U_2X_2Y_1)-\epsilon)}\\
&=&2^{n(S_1-I(X_1;Y_1|U_1U_2X_2Q)+4\epsilon)}=2^{n(S_1-I(X_1;Y_1|U_1X_2Q)+4\epsilon)}
\eqn
\bqn
&&\sum_{k_2\neq 1,l_1=l_2=k_1=1}\Pr(E_{111k_2}^1)\\
&\leq& 2^{nS_2}\sum_{(q^n,u_1^n,u_2^n,x_1^n,x_2^n,y_1^n)\in A_\epsilon^{(n)}}p(x_2^n|u_2^n,q^n)p(q^nu_1^nu_2^nx_1^ny_1^n)\\
&\leq& 2^{nS_2}2^{n(H(QU_1U_2X_1X_2Y_1)+\epsilon)}2^{-n(H(X_2|U_2Q)-2\epsilon)}2^{-n(H(QU_1U_2X_1Y_1)-\epsilon)}\\
&=&2^{n(S_2-I(X_2;Y_1|U_1U_2X_1Q)+4\epsilon)}=2^{n(S_2-I(X_2;Y_1|U_2X_1Q)+4\epsilon)}
\eqn
\bqn
&&\sum_{l_1,l_2\neq 1,k_1=k_2=1}\Pr(E_{l_1l_211}^1)\\
&\leq& 2^{n(T_1+T_2)}\sum_{(q^n,u_1^n,u_2^n,x_1^n,x_2^n,y_1^n)\in A_\epsilon^{(n)}}p(u_1^n,x_1^n,u_2^n,x_2^n|q^n)p(q^ny_1^n)\\
&\leq& 2^{n(T_1+T_2)}2^{n(H(QU_1U_2X_1X_2Y_1)+\epsilon)}2^{-n(H(U_1X_1U_2X_2|Q)-2\epsilon)}2^{-n(H(QY_1)-\epsilon)}\\
&=&2^{n(T_1+T_2-I(U_1X_1U_2X_2;Y_1|Q)+4\epsilon)}=2^{n(T_1+T_2-I(X_1X_2;Y_1|Q)+4\epsilon)}
\eqn
\bqn
&&\sum_{l_1,k_1\neq 1,l_2=k_2=1}\Pr(E_{l_11k_11}^1)\\
&\leq& 2^{n(S_1+T_1)}\sum_{(q^n,u_1^n,u_2^n,x_1^n,x_2^n,y_1^n)\in A_\epsilon^{(n)}}p(u_1^n,x_1^n|q^n)p(q^nu_2^nx_2^ny_1^n)\\
&\leq& 2^{n(S_1+T_1)}2^{n(H(QU_1U_2X_1X_2Y_1)+\epsilon)}2^{-n(H(U_1X_1Q)-2\epsilon)}2^{-n(H(QU_2X_2Y_1)-\epsilon)}\\
&=&2^{n(S_1+T_1-I(U_1X_1;Y_1|U_2X_2Q)+4\epsilon)}=2^{n(S_1+T_1-I(X_1;Y_1|X_2Q)+4\epsilon)}
\eqn
\bqn
&&\sum_{l_1,k_2\neq 1,l_2=k_1=1}\Pr(E_{l_111k_2}^1)\\
&\leq& 2^{n(S_2+T_1)}\sum_{(q^n,u_1^n,u_2^n,x_1^n,x_2^n,y_1^n)\in A_\epsilon^{(n)}}p(u_2^n,x_1^n,x_2^n|u_1^nq^n)p(q^nu_1^ny_1^n)\\
&\leq& 2^{n(S_2+T_1)}2^{n(H(QU_1U_2X_1X_2Y_1)+\epsilon)}2^{-n(H(U_2X_1X_2|U_1Q)-2\epsilon)}2^{-n(H(QU_1Y_1)-\epsilon)}\\
&=&2^{n(S_2+T_1-I(U_2X_1X_2;Y_1|U_1Q)+4\epsilon)}=2^{n(S_2+T_1-I(X_1X_2;Y_1|U_1Q)+4\epsilon)}
\eqn
\bqn
&&\sum_{l_2,k_1\neq 1,l_1=k_2=1}\Pr(E_{1l_2k_11}^1)\\
&\leq& 2^{n(S_1+T_2)}\sum_{(q^n,u_1^n,u_2^n,x_1^n,x_2^n,y_1^n)\in A_\epsilon^{(n)}}p(u_1^n,x_1^n,x_2^n|u_2^nq^n)p(q^nu_2^ny_1^n)\\
&\leq& 2^{n(S_1+T_2)}2^{n(H(QU_1U_2X_1X_2Y_1)+\epsilon)}2^{-n(H(U_1X_1X_2|U_2Q)-2\epsilon)}2^{-n(H(QU_2Y_1)-\epsilon)}\\
&=&2^{n(S_1+T_2-I(U_1X_1X_2;Y_1|U_2Q)+4\epsilon)}=2^{n(S_1+T_2-I(X_1X_2;Y_1|U_2Q)+4\epsilon)}
\eqn
\bqn
&&\sum_{l_2,k_2\neq 1,l_1=k_1=1}\Pr(E_{1l_21k_2}^1)\\
&\leq& 2^{n(S_2+T_2)}\sum_{(q^n,u_1^n,u_2^n,x_1^n,x_2^n,y_1^n)\in A_\epsilon^{(n)}}p(u_2^n,x_2^n|q^n)p(q^nu_1^nx_1^ny_1^n)\\
&\leq& 2^{n(S_2+T_2)}2^{n(H(QU_1U_2X_1X_2Y_1)+\epsilon)}2^{-n(H(U_2X_2|Q)-2\epsilon)}2^{-n(H(QU_1X_1Y_1)-\epsilon)}\\
&=&2^{n(S_2+T_2-I(U_2X_2;Y_1|U_1X_1Q)+4\epsilon)}=2^{n(S_2+T_2-I(X_2;Y_1|X_1Q)+4\epsilon)}
\eqn
\bqn
&&\sum_{k_1,k_2\neq 1,l_1=l_2=1}\Pr(E_{11k_1k_2}^1)\\
&\leq& 2^{n(S_1+S_2)}\sum_{(q^n,u_1^n,u_2^n,x_1^n,x_2^n,y_1^n)\in A_\epsilon^{(n)}}p(x_1^n|u_1^nq^n)p(x_2^n|u_2^nq^n)p(q^nu_1^nu_2^ny_1^n)\\
&\leq& 2^{n(S_1+S_2)}2^{n(H(QU_1U_2X_1X_2Y_1)+\epsilon)}2^{-n(H(X_1X_2|U_1U_2Q)-2\epsilon)}2^{-n(H(QU_1U_2Y_1)-\epsilon)}\\
&=&2^{n(S_1+S_2-I(X_1X_2;Y_1|U_1U_2Q)+4\epsilon)}=2^{n(S_1+S_2-I(X_1X_2;Y_1|U_1U_2Q)+4\epsilon)}
\eqn
\bqn
&&\sum_{l_1,l_2,k_1\neq 1,k_2=1}\Pr(E_{l_1l_2k_11}^1)\\
&\leq& 2^{n(S_1+T_1+T_2)}\sum_{(q^n,u_1^n,u_2^n,x_1^n,x_2^n,y_1^n)\in A_\epsilon^{(n)}}p(u_1^n,x_1^n,u_2^n,x_2^n|q^n)p(q^ny_1^n)\\
&\leq& 2^{n(S_1+T_1+T_2)}2^{n(H(QU_1U_2X_1X_2Y_1)+\epsilon)}2^{-n(H(U_1X_1U_2X_2|Q)-2\epsilon)}2^{-n(H(QY_1)-\epsilon)}\\
&=&2^{n(S_1+T_1+T_2-I(U_1X_1U_2X_2;Y_1|Q)+4\epsilon)}=2^{n(S_1+T_1+T_2-I(X_1X_2;Y_1|Q)+4\epsilon)}
\eqn
\bqn
&&\sum_{l_1,l_2,k_2\neq 1,k_1=1}\Pr(E_{l_1l_21k_2}^1)\\
&\leq& 2^{n(T_1+S_2+T_2)}\sum_{(q^n,u_1^n,u_2^n,x_1^n,x_2^n,y_1^n)\in A_\epsilon^{(n)}}p(u_1^n,x_1^n,u_2^n,x_2^n|q^n)p(q^ny_1^n)\\
&\leq& 2^{n(T_1+S_2+T_2)}2^{n(H(QU_1U_2X_1X_2Y_1)+\epsilon)}2^{-n(H(U_1X_1X_2U_2|Q)-2\epsilon)}2^{-n(H(QY_1)-\epsilon)}\\
&=&2^{n(S_1+T_1+S_2-I(U_1U_2X_1X_2;Y_1|Q)+4\epsilon)}=2^{n(T_1+S_2+T_2-I(X_1X_2;Y_1|Q)+4\epsilon)}
\eqn
\bqn
&&\sum_{l_1,k_1,k_2\neq 1,l_2=1}\Pr(E_{l_11k_1k_2}^1)\\
&\leq& 2^{n(S_1+T_1+S_2)}\sum_{(q^n,u_1^n,u_2^n,x_1^n,x_2^n,y_1^n)\in A_\epsilon^{(n)}}p(u_1^n,x_1^n,x_2^n|u_2^nq^n)p(q^nu_2^ny_1^n)\\
&\leq& 2^{n(S_1+T_1+S_2)}2^{n(H(QU_1U_2X_1X_2Y_1)+\epsilon)}2^{-n(H(U_1X_1X_2|U_2Q)-2\epsilon)}2^{-n(H(QU_2Y_1)-\epsilon)}\\
&=&2^{n(S_1+T_1+S_2-I(U_1X_1X_2;Y_1|U_2Q)+4\epsilon)}=2^{n(S_1+T_1+S_2-I(X_1X_2;Y_1|U_2Q)+4\epsilon)}
\eqn
\bqn
&&\sum_{l_2,k_1,k_2\neq 1,l_1=1}\Pr(E_{1l_2k_1k_2}^1)\\
&\leq& 2^{n(S_1+S_2+T_2)}\sum_{(q^n,u_1^n,u_2^n,x_1^n,x_2^n,y_1^n)\in A_\epsilon^{(n)}}p(u_2^n,x_1^n,x_2^n|u_1^nq^n)p(q^nu_1^ny_1^n)\\
&\leq& 2^{n(S_1+S_2+T_2)}2^{n(H(QU_1U_2X_1X_2Y_1)+\epsilon)}2^{-n(H(U_2X_1X_2|U_1Q)-2\epsilon)}2^{-n(H(QU_1Y_1)-\epsilon)}\\
&=&2^{n(S_1+S_2+T_2-I(U_2X_1X_2;Y_1|U_1Q)+4\epsilon)}=2^{n(S_1+S_2+T_2-I(X_1X_2;Y_1|U_1Q)+4\epsilon)}
\eqn
\bqn
&&\sum_{l_1,l_2,k_1,k_2\neq 1}\Pr(E_{l_1l_2k_1k_2}^1)\\
&\leq& 2^{n(S_1+T_1+S_2+T_2)}\sum_{(q^n,u_1^n,u_2^n,x_1^n,x_2^n,y_1^n)\in A_\epsilon^{(n)}}p(u_1^n,x_1^n,u_2^n,x_2^n|q^n)p(q^ny_1^n)\\
&\leq& 2^{n(S_1+T_1+S_2+T_2)}2^{n(H(QU_1U_2X_1X_2Y_1)+\epsilon)}2^{-n(H(U_1X_1U_2X_2|Q)-2\epsilon)}2^{-n(H(QY_1)-\epsilon)}\\
&=&2^{n(S_1+T_1+S_2+T_2-I(U_1U_2X_1X_2;Y_1|Q)+4\epsilon)}=2^{n(S_1+T_1+S_2+T_2-I(X_1X_2;Y_1|Q)+4\epsilon)}
\eqn
Putting them together, we have
\bqn
P_{e1}^n&\leq&\epsilon+2^{n(T_1-I(X_1;Y_1|X_2Q)+4\epsilon)}+2^{n(T_2-I(X_2;Y_1|X_1Q)+4\epsilon)}\\
&&+2^{n(T_1-I(X_1;Y_1|U_1X_2Q)+4\epsilon)}
+2^{n(S_2-I(X_2;Y_1|U_2X_1Q)+4\epsilon)}\\
&&+2^{n(T_1+T_2-I(X_1X_2;Y_1|Q)+4\epsilon)}
+2^{n(S_1+T_1-I(X_1;Y_1|X_2Q)+4\epsilon)}\\
&&+2^{n(S_2+T_1-I(X_1X_2;Y_1|U_1Q)+4\epsilon)}
+2^{n(S_1+T_2-I(X_1X_2;Y_1|U_2Q)+4\epsilon)}\\
&&+2^{n(S_2+T_2-I(X_2;Y_1|X_1Q)+4\epsilon)}
+2^{n(S_1+S_2-I(X_1X_2;Y_1|U_1U_2Q)+4\epsilon)}\\
&&+2^{n(S_1+T_1+T_2-I(X_1X_2;Y_1|Q)+4\epsilon)}
+2^{n(T_1+S_2+T_2-I(X_1X_2;Y_1|Q)+4\epsilon)}\\
&&+2^{n(S_1+T_1+S_2-I(X_1X_2;Y_1|U_2Q)+4\epsilon)}
+2^{n(S_1+S_2+T_2-I(X_1X_2;Y_1|U_1Q)+4\epsilon)}\\
&&+2^{n(S_1+T_1+S_2+T_2-I(X_1X_2;Y_1|Q)+4\epsilon)}
\eqn
For user $2$, we
define the following event:
\bqa
E_{l_1l_2m_3}^2=\left\{(q^n,u_1^n(l_1),u_2^n(l_2),x_3^n(m_3),y_2^n)\in A_\epsilon^{(n)}(QU_1U_2X_3Y_2)\right\}.
\eqa
The error probability at receiver $2$ is
\bqn
P_{e2}^n&=&\Pr\left\{{E_{111}^2}^c\bigcup\cup_{m_3\neq 1,any(l_1,l_2)}E_{l_1l_2m_3}^2\right\}\\
&\leq&\Pr\left({E_{111}^2}^c\right)+\sum_{m_3\neq 1, l_1=l_2=1}\Pr\left(E_{11m_3}^2\right)+\sum_{l_1,m_3\neq 1, l_2=1}\Pr\left(E_{l_11m_3}^2\right)\\&&+\sum_{l_2,m_3\neq 1, l_1=1}\Pr\left(E_{1l_2m_3}^2\right)+\sum_{l_1,l_2,m_3\neq 1}\Pr\left(E_{l_1l_2m_3}^2\right)
\eqn

Again, it is obvious that $\Pr({E_{111}^2}^c)\rightarrow 0$ when $n\rightarrow \infty$. From the joint typicality we have
\bqn
\sum_{m_3\neq 1, l_1=l_2=1}\Pr\left(E_{11m_3}^2\right)&\leq&2^{nR_3}\sum_{(q^n,u_1^n,u_2^n,x_3^n,y_2^n)\in A_\epsilon^{(n)}}p(x_3^n|q^n)p(q^n,u_1^n,u_2^n,y_2^n)\\
&\leq&2^{nR_3}2^{n(H(QU_1U_2X_3Y_2)+\epsilon)}2^{-n(H(X_3|Q)-2\epsilon)}2^{-n(H(QU_1U_2Y_2)-\epsilon)}\\
&=&2^{n(R_3-I(X_3;Y_2|U_1U_2Q)+4\epsilon)}
\eqn
\bqn
\sum_{l_1,m_3\neq 1,l_2=1}\Pr\left(E_{l_11m_3}^2\right)&\leq&2^{n(T_1+R_3)}\sum_{(q^n,u_1^n,u_2^n,x_3^n,y_2^n)\in A_\epsilon^{(n)}}p(u_1^n,x_3^n|q^n)p(q^n,u_2^n,y_2^n)\\
&\leq&2^{n(T_1+R_3)}2^{n(H(QU_1U_2X_3Y_2)+\epsilon)}2^{-n(H(U_1,X_3|Q)-2\epsilon)}2^{-n(H(QU_2Y_2)-\epsilon)}\\
&=&2^{n(T_1+R_3-I(U_1X_3;Y_2|U_2Q)+4\epsilon)}
\eqn
\bqn
\sum_{l_2,m_3\neq 1,l_1=1}\Pr\left(E_{1l_2m_3}^2\right)&\leq&2^{n(T_2+R_3)}\sum_{(q^n,u_1^n,u_2^n,x_3^n,y_2^n)\in A_\epsilon^{(n)}}p(u_2^n,x_3^n|q^n)p(q^n,u_1^n,y_2^n)\\
&\leq&2^{n(T_2+R_3)}2^{n(H(QU_1U_2X_3Y_2)+\epsilon)}2^{-n(H(U_2,X_3|Q)-2\epsilon)}2^{-n(H(QU_1Y_2)-\epsilon)}\\
&=&2^{n(T_2+R_3-I(U_2X_3;Y_2|U_1Q)+4\epsilon)}
\eqn
\bqn
\sum_{l_1,l_2,m_3\neq 1}\Pr\left(E_{l_1l_2m_3}^2\right)&\leq&2^{n(T_1+T_2+R_3)}\sum_{(q^n,u_1^n,u_2^n,x_3^n,y_2^n)\in A_\epsilon^{(n)}}p(u_1^n,u_2^n,x_3^n|q^n)p(q^n,y_2^n)\\
&\leq&2^{n(T_1+T_2+R_3)}2^{n(H(QU_1U_2X_3Y_2)+\epsilon)}2^{-n(H(U_1U_2X_3|Q)-2\epsilon)}2^{-n(H(QY_2)-\epsilon)}\\
&=&2^{n(T_1+T_2+R_3-I(U_1U_2X_3;Y_2|Q)+4\epsilon)}
\eqn
Therefore, for receiver $2$,
\bqn
P_{e2}^n&\leq&\epsilon+2^{n(R_3-I(X_3;Y_2|U_1U_2Q)+4\epsilon)}
+2^{n(T_1+R_3-I(U_1X_3;Y_2|U_2Q)+4\epsilon)}\\
&&+2^{n(T_2+R_3-I(U_2X_3;Y_2|U_1Q)+4\epsilon)}
+2^{n(T_1+T_2+R_3-I(U_1U_2X_3;Y_2|Q)+4\epsilon)}
\eqn
In order that $P_{e1}^n$, $P_{e2}^n$ $\rightarrow 0$, from above inequalities, we must have
\bqa
T_1&\leq&I(X_1;Y_1|X_2Q),\label{eq:original_start}\\
T_2&\leq&I(X_2;Y_1|X_1Q),\\
T_1&\leq&I(X_1;Y_1|U_1X_2Q),\\
S_2&\leq&I(X_2;Y_1|U_2X_1Q),\\
T_1+T_2&\leq&I(X_1X_2;Y_1|Q),\\
S_1+T_1&\leq&I(X_1;Y_1|X_2Q),\\
S_2+T_1&\leq&I(X_1X_2;Y_1|U_1Q),\\
S_1+T_2&\leq&I(X_1X_2;Y_1|U_2Q),\\
S_2+T_2&\leq&I(X_2;Y_1|X_1Q),\\
S_1+S_2&\leq&I(X_1X_2;Y_1|U_1U_2Q),\\
S_1+T_1+T_2&\leq&I(X_1X_2;Y_1|Q),\\
T_1+S_2+T_2&\leq&I(X_1X_2;Y_1|Q),\\
S_1+T_1+S_2&\leq&I(X_1X_2;Y_1|U_2Q),\\
S_1+S_2+T_2&\leq&I(X_1X_2;Y_1|U_1Q),\\
S_1+T_1+S_2+T_2&\leq&I(X_1X_2;Y_1|Q),\\
R_3&\leq&I(X_3;Y_2|U_1U_2Q),\\
T_1+R_3&\leq&I(U_1X_3;Y_2|U_2Q),\\
T_2+R_3&\leq&I(U_2X_3;Y_2|U_1Q),\\
T_1+T_2+R_3&\leq&I(U_1U_2X_3;Y_2|Q).\label{eq:original_end}
\eqa
Using the Fourier-Motzkin elimination on (\ref{eq:original_start})-(\ref{eq:original_end}) and getting rid of redundant inequalities, we obtain (\ref{eq:thm1_start})-(\ref{eq:thm1_end}). The cardinality bounds on the auxiliary random variables are from the Caratheodory Theorem.
\subsection{Proof of Theorem \ref{thm:discrete_strong}}\label{apx:strong_cr}
The achievability part follows directly from Theorem \ref{thm:general_inner} by setting $U_1=U_2=\emptyset$.
For the converse, (\ref{eq:cr_strong_1}), (\ref{eq:cr_strong_2}) and (\ref{eq:cr_strong_4}) form an outer bound on the capacity region of the corresponding MAC with $X_1$ and $X_2$ as inputs and $Y_1$ as output. Moreover, (\ref{eq:cr_strong_3}) is a natural bound on $R_3$. Therefore, we only need to prove (\ref{eq:cr_strong_5})-(\ref{eq:cr_strong_7}).
First,
\bqn
n(R_2+R_3)-n\epsilon&=&H(W_2)+H(W_3)-n\epsilon\\
&\stackrel{(a)}{\leq}&I(X_2^n;Y_1^n)+I(X_3^n;Y_2^n)\\
&\stackrel{(b)}{\leq}&I(X_2^n;Y_1^n|X_1^n)+I(X_3^n;Y_2^n|X_1^n)\\
&\stackrel{(c)}{\leq}&I(X_2^n;Y_2^n|X_1^nX_3^n)+I(X_3^n;Y_2^n|X_1^n)\\
&=&I(X_2^nX_3^n;Y_2^n|X_1^n)\\
&=&H(Y_2^n|X_1^n)-H(Y_2^n|X_1^nX_2^nX_3^n)\\
&=&\sum_{i=1}^n\left\{H(Y_{2i}|Y_2^{i-1}X_1^n)-H(Y_{2i}|Y_2^{i-1}X_1^nX_2^nX_3^n)\right\}\\
&\stackrel{(d)}{\leq}&\sum_{i=1}^n\left\{H(Y_{2i}|X_{1i})-H(Y_{2i}|X_{1i}X_{2i}X_{3i})\right\}\\
&=&I(X_{2i}X_{3i};Y_{2i}|X_{1i}),
\eqn
where $(a)$ is from Fano's inequality; $(b)$ is because of the mutual independence among $X_1^n$, $X_2^n$ and $X_3^n$; $(c)$ is due to (\ref{eq:multi_strong_2}); and $(d)$ uses the fact that conditioning reduces entropy and the memoryless property.
Similarly, we can prove the bound on $R_1+R_3$. We further have
\bqn
n(R_1+R_2+R_3)-n\epsilon&=&H(W_1,W_2)+H(W_3)-n\epsilon\\
&\stackrel{(a)}{\leq}&I(X_1^nX_2^n;Y_1^n)+I(X_3^n;Y_2^n)\\
&\stackrel{(b)}{\leq}&I(X_1^nX_2^n;Y_1^n)+I(X_3^n;Y_2^n)\\
&\stackrel{(c)}{\leq}&I(X_1^nX_2^n;Y_2^n|X_3^n)+I(X_3^n;Y_2^n)\\
&=&I(X_1^nX_2^nX_3^n;Y_2^n)\\
&=&H(Y_2^n)-H(Y_2^n|X_1^nX_2^nX_3^n)\\
&=&\sum_{i=1}^n \left\{H(Y_{2i}|Y_2^{i-1})-H(Y_{2i}|Y_2^{i-1}X_1^nX_2^nX_3^n)\right\}\\
&\stackrel{(d)}{\leq}&\sum_{i=1}^n \left\{H(Y_{2i}-H(Y_{2i}|X_{1i}X_{2i}X_{3i}))\right\}\\
&=&I(X_{1i}X_{2i}X_{3i};Y_{2i}).
\eqn
By introducing a time-sharing random variable $Q$, we obtain Theorem \ref{thm:discrete_strong}. The cardinality of $\Qmat$ can be verified using the Caratheodory theorem.
\end{appendix}

\end{document}